\documentclass[11pt,a4paper]{article}

\usepackage[margin=1in]{geometry}

\usepackage{amsmath,amssymb,amsthm,mathtools}
\usepackage{physics}
\usepackage{bm}
\usepackage{siunitx}

\usepackage{graphicx}
\usepackage{subcaption}
\usepackage{booktabs}
\usepackage{xcolor}
\usepackage[colorlinks=true,linkcolor=blue,citecolor=blue,urlcolor=blue]{hyperref}
\usepackage[nameinlink,capitalize]{cleveref}
\usepackage{floatrow}   
\graphicspath{{./images/}}

\usepackage[numbers,sort&compress]{natbib}

\newtheorem{proposition}{Proposition}
\theoremstyle{definition}
\newtheorem{definition}{Definition}

\DeclareMathOperator{\ad}{ad}
\DeclareMathOperator{\Hom}{Hom}
\newcommand{\Pexp}[0]{\mathcal{P}\exp}

\title{\textbf{Tropical WKB asymptotics of NRS coordinates for opers in $SU(2)$, $N_f=4$ theory}}
\author{
  Vasilii Iugov\\
  \small YITP, Stony Brook University, Stony Brook, US\\
  vasilii.iugov@stonybrook.edu
}
\date{}

\begin{document}

\maketitle

\begin{abstract}
We study the semiclassical limit of $SL_2$-opers on the four-punctured sphere in Nekrasov--Rosly--Shatashvili Darboux coordinates. Using exact Wentzel--Kramers--Brillouin (WKB) connection formulae, we express the trace coordinates of the corresponding $SL_2(\mathbb C)$ character variety as finite Laurent sums of Voros exponentials. Tropicalizing these formulae and the NRS relations gives a chamberwise integer affine linear system for the leading logarithms of the NRS coordinates. In chambers where this system is unimodular and the selected cycles form a primitive symplectic pair, the leading asymptotics agree, up to flavor-period shifts, with Seiberg--Witten periods of the $\mathcal N=2$ $SU(2)$ theory with $N_f=4$ fundamental hypermultiplets. We verify this mechanism in a sample chamber and in the weak-coupling degeneration. No global coordinate-independent recovery theorem is claimed; non-unimodular or degenerate chambers are treated as limitations of the chosen NRS chart. In the weak-coupling degeneration, we show that the NRS chart can be chosen compatibly with the plumbing limit so that the resulting chamber is unimodular and non-degenerate away from tropical walls.
\end{abstract}

\section{Introduction}
\label{sec:introduction}

Seiberg-Witten geometry appears in several closely related but technically
different forms. On one side, it is described by periods of a meromorphic
differential on a family of curves over the Coulomb branch. On another, for class \(S\) theories,
the same geometry is encoded in moduli spaces of flat \(SL_2(\mathbb C)\)
connections, where opers form a distinguished Lagrangian subvariety. Relating
these descriptions requires understanding how the monodromy data of an oper
develops WKB asymptotics controlled by the Seiberg--Witten differential in the semiclassical limit. This is the problem addressed in this
paper, with particular calculations done for $A_1$ theory on a 4-punctured sphere.

The physical motivation for this question comes from comparing two different
deformations of the same four-dimensional \(\mathcal N=2\) supersymmetric gauge
theory. One deformation is obtained by compactifying the theory on
\(\mathbb R^3\times S^1_R\). Following Seiberg and Witten, Gaiotto, Moore and
Neitzke studied the resulting three-dimensional low-energy theory, whose
target space is the Hitchin moduli space equipped with an \(R\)-dependent
hyperkähler metric \cite{gaiotto_four-dimensional_2010,gaiotto_wall-crossing_2011}.
The metric is described by a twistor family of holomorphic symplectic
structures and Darboux coordinates. In the large-\(R\) limit, it approaches the
semiflat metric determined by the Seiberg--Witten period map, while the
finite-\(R\) corrections are controlled by BPS states and wall-crossing.

Another deformation is the \(\Omega\)-deformation of the four-dimensional
theory. Nekrasov introduced the \(\Omega\)-deformation on \(\mathbb R^4\)
\cite{nekrasov_seiberg-witten_2002}, and Nekrasov and Shatashvili later studied
the limit in which only one equivariant parameter remains,
\[
    \epsilon_1=\hbar,\qquad \epsilon_2=0 .
\]
In this limit the theory is effectively two-dimensional and has an effective
twisted superpotential. The expectation in Nekrasov's approach is that, as
\(\hbar\to0\), this twisted superpotential approaches the Seiberg--Witten
prepotential in the form
\[
    \widetilde W_{\mathrm{eff}} \sim \frac{1}{\hbar}\mathcal F,
\]
up to convention-dependent linear and logarithmic terms. Thus the flat-space
limit of the \(\Omega\)-deformed theory should reproduce the same
Seiberg--Witten geometry that appears in the large-\(R\) limit of the
GMN construction.

The geometry of this identification is not completely tautological, because
the two constructions describe different aspects of the same hyperkähler
manifold. GMN use WKB triangulations and Fock--Goncharov-type coordinates to
construct Darboux coordinates in the twistor family. In the
\(\Omega\)-deformed setup, the relevant object is instead the brane of opers
inside the same moduli space of flat connections. The oper brane has a natural
description in Nekrasov--Rosly--Shatashvili coordinates, which are adapted to a
pants decomposition and are expressed in terms of monodromy traces
\cite{nekrasov_darboux_2013}. There is no equally direct description of this
oper Lagrangian in the Fock--Goncharov coordinates used by GMN. Thus the two
approaches explore the same hyperkähler manifold, but in different coordinate
systems and in different physical limits.

The purpose of this paper is to check the consistency of these two flat-space
limits. We study the brane of opers in NRS coordinates and take the
semiclassical limit \(\hbar\to0\). The tropical limit is chamberwise, but not every chamber has the same interpretation. We therefore distinguish integral-period chambers, in which the dominant WKB asymptotics of the NRS Darboux coordinates form an integral symplectic pair of Seiberg--Witten periods, from non-integral-period chambers, in which the same tropical procedure gives a well-defined leading asymptotic but not an integral period basis in the chosen NRS chart.

\paragraph{Seiberg--Witten geometry.}
In \(4d\), the Coulomb-branch low-energy dynamics of many \(\mathcal N=2\)
theories is determined by a curve \(\Sigma\), known as the Seiberg--Witten
curve, and a meromorphic differential
\(\lambda\in H^0(\Sigma,K_\Sigma(D))\) for an appropriate pole divisor \(D\).
Coordinates on the moduli space of vacua of the theory are given by periods
\[
    a_i=\int_{A_i}\lambda,\qquad
    a_D^i=\int_{B_i}\lambda
\]
\cite{seiberg_monopole_1994,seiberg_monopoles_1994,Freed_1999}.

More invariantly, for a pure gauge theory the Coulomb branch \(B\) is the base
of a bundle \(\Gamma\) of electric-magnetic charge lattices with an integral
antisymmetric pairing \(\langle\cdot,\cdot\rangle\), and a holomorphic central
charge function \(Z:\Gamma\to\mathbb C\). In a local electric-magnetic duality
frame, the electric central charges give special coordinates \(a_i\), while the
magnetic central charges are
\[
    a_{D,i}=\frac{\partial\mathcal F}{\partial a_i}
\]
for the prepotential \(\mathcal F\). BPS particles of charge \(\gamma\) have
masses bounded below by \(|Z_\gamma|\), with equality for BPS states.

There are several known classes of \(4d\) \(\mathcal N=2\) theories. Class \(S\)
theories are defined by a choice of ADE Lie algebra \(\mathfrak g\), a curve
\(\mathcal C\), and data at the punctures of \(\mathcal C\). The \(4d\) theory
is obtained by a twisted compactification of a \(6d\) \(\mathcal N=(2,0)\)
theory on the curve \(\mathcal C\) \cite{gaiotto_n2_2012}. The
Seiberg--Witten curve of the resulting \(4d\) theory is the spectral curve of
the Hitchin system on \(\mathcal C\), which is a branched cover
\(\Sigma\subset T^*\mathcal C\)
\cite{hitchin_self-duality_1987,hitchin_1987,gaiotto_wall-crossing_2011}.
The structure of a particular cover is determined by meromorphic differentials
\[
    \phi_k=P_k(\Phi)\in H^0(\mathcal C,K_{\mathcal C}^k(D_k)),
\]
where \(\Phi\) is the Hitchin field, \(P_k\) are invariant polynomials of the
Lie algebra and \(D_k\) are the allowed poles. The non-singular coefficients of
\(\phi_k\) form a set of coordinates on the Coulomb branch, while their
puncture data encodes mass parameters and flavor symmetries.

For Lagrangian \(\mathcal N=2\) theories, the Seiberg--Witten curves were
originally deduced from the weak-coupling behavior of the theory, as well as
monodromies of central charges around cycles and singularities where a
particle becomes massless
\cite{seiberg_monopole_1994,seiberg_monopoles_1994}. Many of these curves were
also given a later geometric interpretation, for example as spectral curves of
algebraic integrable systems
\cite{Donagi_1996,donagi_seiberg-witten_1997,takasaki_spectral_1997,Martinec_1996,Gorsky_1995,D_Hoker_1998},
M-theory fivebrane curves \cite{Witten_1997,klemm_self-dual_1996}, or special
cases and limits of class \(S\) Hitchin systems
\cite{gaiotto_n2_2012,gaiotto_wall-crossing_2011}.

\paragraph{Flat connections.}
After compactification on a circle, the low-energy theory is a
three-dimensional sigma model whose target \(\mathcal M\) is the Hitchin
moduli space. It is fibered over the four-dimensional Coulomb branch \(B\). We
restrict to the \(A_1\) case, where the relevant Hitchin moduli space can be
described in complex structure \(J\) as a moduli space of flat
\(SL(2,\mathbb C)\)-connections with fixed conjugacy classes around punctures,
modulo gauge transformations. Those are in one-to-one correspondence with
\(SL(2,\mathbb C)\) representations of \(\pi_1(\mathcal C)\) with
fixed conjugacy classes at the punctures \(P\), modulo total conjugation:
\[
    \mathcal M
    =
    \{\rho:\pi_1(\mathcal C)\to SL(2,\mathbb C)
    \text{ with fixed puncture conjugacy classes}\}/SL(2,\mathbb C).
\]
The symplectic structure holomorphic with respect to the complex structure
\(J\) on \(\mathcal M\) is the Atiyah--Bott--Goldman symplectic form
\[
    \Omega_J=\int_{\mathcal C}\operatorname{Tr}(\delta\mathcal A\wedge\delta\mathcal A).
\]

There are several notable choices of Darboux coordinates on \(\mathcal M\).
Fock--Goncharov coordinates are associated to WKB triangulations and play a
central role in the GMN construction of canonical Darboux coordinates
\cite{gaiotto_wall-crossing_2011}. By contrast, the NRS coordinates are
adapted to a pants decomposition and are expressed directly in terms of
monodromy traces \cite{nekrasov_darboux_2013}. Our analysis concerns the WKB
asymptotics of the NRS coordinates, not the WKB triangulations used in the GMN
construction.

\paragraph{Brane of opers.}
Consider the \(4d\) theory on
\[
    X_4=C_{\mathrm{cigar}}\times S^1\times\mathbb R,
\]
with \(C_{\mathrm{cigar}}\) a two-dimensional manifold with cigar metric
\[
    g_{\mathrm{cigar}}=dr^2+f(r)^2d\phi^2 .
\]
Here \(f(r)\sim r\) as \(r\to0\), while \(f(r)\to R_{\mathrm{cigar}}\) as
\(r\to\infty\). Introduce an \(\Omega\)-deformation
\cite{nekrasov_seiberg-witten_2002} parameter \(\hbar\) for the cigar isometry
\(\partial_\phi\). Upon compactification on
\(S^1_{\mathrm{cigar}}\times S^1\), one gets a two-dimensional sigma model on
\(\mathbb R_+\times\mathbb R\) with target space \(\mathcal M\). In a
particular choice of the duality frame, the boundary at the tip of the cigar
was identified with the brane of opers \(B_{\mathcal O_\tau}\), corresponding
to a Lagrangian submanifold
\[
    \mathcal O_\tau\subset\mathcal M
\]
\cite{nekrasov_omega_2010,beilinson_opers_2005,kapustin2007electricmagneticdualitygeometriclanglands}.
The boundary condition at infinity becomes another Lagrangian brane
\(B_\gamma^\infty\), independent of the four-dimensional coupling \(\tau\).
The Lagrangian variety \(\mathcal O_\tau\subset\mathcal M\) is the space of
flat connections that are equivalent, in some local trivialization, to a
second-order meromorphic differential operator of the form
\[
    -\partial^2+\frac{1}{\hbar^2}T(z; U).
\]

Another way to look at this system is through the lens of an
\(\mathcal N=(2,2)\) two-dimensional theory on \(S^1\times\mathbb R\). The
\(\Omega\)-deformation parameter \(\hbar\) is the twisted mass for the
\(U(1)\) rotation symmetry of the cigar. At low energy, the theory has an
effective twisted superpotential
\[
    \widetilde W_{\mathrm{eff}}(a,\tau,\mathfrak m,\hbar,\gamma)
    =
    \hbar\left(
    W_{\mathcal O_\tau}
    \left(\frac{a}{\hbar},\frac{\mathfrak m}{\hbar}\right)
    -
    W_\infty
    \left(\frac{a}{\hbar},\frac{\mathfrak m}{\hbar},\gamma\right)
    \right),
\]
as proposed in \cite{nekrasov_darboux_2013} and proved in
\cite{jeong_opers_2018}. The generating function of the oper variety can be
expressed as a limit of the fully \(\Omega\)-deformed four-dimensional theory:
\[
    W_{\mathcal O_\tau}(a,\nu)
    =
    \frac{1}{\hbar}
    \lim_{\epsilon_2\to0}
    \epsilon_2\log
    Z(a\hbar,\tau,\nu\hbar;\hbar,\epsilon_2).
\]
In the limit \(\hbar=\epsilon_1\to0\), \(\epsilon_2=0\), the monodromy
variables
\[
    \alpha=\frac{a}{\hbar},\qquad
    \nu=\frac{\mathfrak m}{\hbar}
\]
go to infinity. Therefore, the limit \(\hbar\to0\) does not probe a compact
region of the oper Lagrangian
\(\mathcal O_\tau\subset\mathcal M\). It probes an asymptotic region in which
the monodromy parameters become large. As we will study below, the particular
behavior depends on the values of \(a\) and \(\mathfrak m\) in a tropical way.

\paragraph{The \(SU(2)\), \(N_f=4\) dictionary.}
The example studied in this paper is the \(A_1\) class \(S\) theory associated
to a four-punctured sphere
\[
    \mathcal C=\mathbb{CP}^1\setminus\{z_1,z_2,z_3,z_4\}.
\]
This theory is the four-dimensional \(\mathcal N=2\) \(SU(2)\) gauge theory
with four fundamental hypermultiplets. The complex structure modulus of the
four-punctured sphere is the exactly marginal gauge coupling of the
four-dimensional theory. In a weak-coupling channel we may choose coordinates
\[
    (z_1,z_2,z_3,z_4)=(0,q,1,\infty),
    \qquad q=e^{2\pi i\tau_{UV}},
\]
so that the degeneration \(q\to0\) corresponds to the weak-coupling limit in
the pants decomposition separating \(\{z_1,z_2\}\) from \(\{z_3,z_4\}\).

The conjugacy class of the monodromy around the puncture \(z_i\) is fixed by
\[
    \operatorname{Tr} g_i=\mu_i=2\cos(2\pi\nu_i).
\]
In the semiclassical limit of the oper equation we keep
\[
    \mathfrak m_i=\hbar\nu_i
\]
fixed. These parameters are the residues of the Seiberg--Witten differential
at the punctures. The four physical hypermultiplet masses are linear
combinations of these residue parameters; in the convention used below one
may take
\[
    M_1=\mathfrak m_1+\mathfrak m_2,\qquad
    M_2=\mathfrak m_1-\mathfrak m_2,\qquad
    M_3=\mathfrak m_3+\mathfrak m_4,\qquad
    M_4=\mathfrak m_3-\mathfrak m_4,
\]
up to the usual signs, permutations, and normalization conventions for the
\(SO(8)\) flavor weights. Thus the puncture data encode the mass parameters,
while the single complex modulus \(q\) encodes the UV gauge coupling.

The remaining complex modulus of the Seiberg--Witten curve is the Coulomb
branch parameter \(U\). In the WKB limit the oper
\[
    \left(\partial_z^2-\hbar^{-2}T(z; U)\right)\psi=0
\]
degenerates to the Seiberg--Witten curve
\[
    \Sigma_U:\quad p^2=T(z; U),
    \qquad
    \lambda=p\,dz.
\]
The periods of \(\lambda\) give the electric and magnetic special coordinates
\[
    a=\oint_{\gamma}\lambda,\qquad
    a_D=\oint_{\gamma_D}\lambda.
\]
The main point of the paper is then that the WKB asymptotics of the NRS
coordinates \((\alpha,\beta)\) on the oper Lagrangian admit a tropical chamber
decomposition. In favorable chambers these asymptotics reproduce an integral
symplectic pair of Seiberg--Witten periods, up to flavor-period shifts. In other
chambers the same tropical extraction can still be well defined, but the chosen
NRS chart need not produce an integral electric-magnetic period basis.

\paragraph{Scope and main technical output.}
The result of this paper is deliberately chamberwise. Fix a generic Stokes
chamber of the \(SL_2\)-oper
\[
    \left(\partial_z^2-\hbar^{-2}T(z)\right)\psi=0,
    \qquad \hbar\to0,
\]
a ray \(\arg\hbar\), and branches of the relevant logarithms. The exact-WKB
monodromy formulae express the trace functions as finite Laurent sums of Voros
exponentials. In a tropical subchamber with unique dominant terms and no
leading cancellation, the NRS relations reduce to an integer affine linear
system for the leading logarithms of \(\alpha_{\mathrm{NRS}}\) and
\(\beta_{\mathrm{NRS}}\).

The paper computes this mechanism explicitly for the generic Stokes graph types
and uses it to identify integral-period chambers in which the linear system is
unimodular, modulo flavor-period shifts, and the selected cycles form a
primitive symplectic pair on the Seiberg--Witten curve
\[
    \lambda=p\,dz,\qquad \Sigma_U: p^2=T(z;U).
\]
In such chambers,
\[
    \alpha_{\mathrm{NRS}}
    =
    \frac{1}{\hbar}\int_{\gamma_\alpha}\lambda+O(1),
    \qquad
    \beta_{\mathrm{NRS}}
    =
    \frac{1}{\hbar}\int_{\gamma_\beta}\lambda+O(1),
    \qquad
    \langle\gamma_\alpha,\gamma_\beta\rangle=\pm1.
\]
We do not prove that every tropical chamber has this property. We call a
subchamber non-integral-period if the same tropical extraction is well defined
but the resulting leading logarithms are rational, non-primitive, degenerate,
dominated by flavor-period data, or otherwise fail to give an integral
symplectic period basis in the chosen NRS chart.

Crossing a tropical wall changes the dominant Voros exponentials and hence can
replace \((\gamma_\alpha,\gamma_\beta)\) by another integral combination of
period data, or can move the asymptotics into a non-integral-period chamber.

\paragraph{Relation to wall-crossing.}
The WKB method used here is parallel to the GMN WKB analysis, but applied to a
different set of Darboux coordinates. The semiclassical limit here is different
in origin: it is the \(\Omega\)-deformation parameter \(\hbar\), instead of the
inverse-radius/twistor parameter in the circle compactification used in
\cite{gaiotto_four-dimensional_2010}. GMN use WKB triangulations to construct
canonical Fock--Goncharov coordinates \(X_\gamma\) on the Hitchin moduli space
and study their Kontsevich--Soibelman wall-crossing. Here we instead use exact
WKB monodromy formulae for \(SL_2\)-opers and then express the resulting
monodromy traces in NRS coordinates. In our approach, the Coulomb-branch point
and the Hitchin spectral curve are obtained in the limit \(\hbar\to0\), while
in the works of Gaiotto et al. the Hitchin spectral curve is the starting
point.

Thus the comparison is not an equality of coordinate systems. Rather, both
constructions probe the same hyperkähler moduli space and must reduce, in
their respective flat-space limits, to the same Seiberg--Witten period
geometry. The calculation below verifies this statement directly for
\(SL_2\)-opers on the four-punctured sphere.

\paragraph{Acknowledgements.} I would like to express my deep gratitude to Professor Nikita Nekrasov, who introduced me
to this field of study and guided me through the research. Research was supported by NSF PHY Award 2310279 and by the Simons Collaboration ``Probabilistic Paths to QFT''.

\section{Flat $SL_2(\mathbb C)$-connections on the four-punctured sphere}
Let $G$ be a connected complex Lie group and $\mathfrak g$ its Lie algebra, equipped with a non-degenerate invariant symmetric bilinear form
$\ev{\bullet,\bullet}$. Consider a principal $G$-bundle $\mathcal P$ over a Riemann surface $\mathcal C$ of genus $g$ with $n$ punctures. 
Any connection $\nabla$ on $\mathcal P$ in local coordinates can be written in the form $\nabla_A = \dd +A$, where $A$ is a differential $1$-form with values in $\mathfrak g$. $A$ is defined on the local chart and depends on a choice of local trivialization of $P$. 
Under a change of local trivialization, $A$ transforms as
\begin{equation}
	\nabla_A \mapsto g^{-1} \nabla_A g, \quad 
	A \mapsto g^{-1} A g + g^{-1} \dd g .
\end{equation}

The curvature of the connection $\nabla_A$ is a $2$-form $F_A$ with values in the adjoint bundle $\ad P$ defined as
\begin{equation}
	F_A = \nabla_A^2 = \dd A + A \wedge A 
	= \frac12(\partial_i A_j - \partial_j A_i +[A_i, A_j]) 
	\dd x^i \wedge \dd x^j .
\end{equation}

Now we can define the moduli space $\mathcal M_{g,n;\nu}$ of flat $G$-connections on the punctured surface 
$\mathcal C$. 
On its smooth locus, it consists of connections with $F_A=0$ and fixed conjugacy classes of holonomies around the punctures, modulo the action of the gauge group $\mathcal G$:
\begin{equation}
	\mathcal M_{g,n;\nu} = 
	\left\{ \nabla_A: F_A = 0, \left[\Pexp\left(\oint_{z_i} A\right)\right] = \nu_i\right\}/\mathcal G .
\end{equation}
Here $z_i$ are the punctures of $\mathcal C$, the integral is taken around small counterclockwise loops enclosing the punctures, and $[\bullet]$ denotes the conjugacy class of the path-ordered exponential.

There is a one-to-one correspondence between flat connections modulo gauge transformations and $\nu$-constrained representations of the fundamental group 
$\pi_1(\mathcal C)$ up to conjugation:
\begin{equation}
	\mathcal M_{g,n;\nu} = \Hom(\pi_1(\mathcal C), G)^\nu / G .
\end{equation}
One can parametrize a point in $\mathcal M_{g,n;\nu}$ by fixing the monodromies around a set of generators of the fundamental group, such that the monodromies around the punctures $z_i$ have the prescribed conjugacy classes $\nu_i$. 

Thus, we have defined the moduli space of flat connections as a character variety. Strictly speaking, it can have singularities; throughout this paper we work on the smooth locus where the symplectic form is non-degenerate.
It also has a holomorphic symplectic structure. Recall the symplectic form $\Omega$ on the space of all $G$-connections on $\Sigma$:
\begin{equation}
	\Omega = \frac12\int_{\Sigma} \ev{\delta A \wedge \delta A}.
\end{equation}

It is easy to check that $\Omega$ is invariant with respect to the action of $\mathcal G$. 
One convenient way to describe the reduction is to remove small disks around the punctures and work with a surface with boundary. Let \(\mathcal G_0\) be the group of gauge transformations that restrict to the identity on the boundary circles.
The first step is to consider the symplectic quotient by the group $\mathcal G_0$. The moment map is the curvature of the connection, so the zero level of the moment map is the space of flat connections. After quotienting by $\mathcal G_0$, one obtains the moduli space of framed flat connections, where the framings are placed at the boundary circles around the punctures.

The remaining quotient by the boundary gauge transformations changes the framings and acts by conjugation on the boundary monodromies. Fixing the conjugacy classes of these monodromies selects a symplectic leaf of the resulting Poisson variety. This symplectic leaf is precisely \(\mathcal M_{g,n;\nu}\).
This construction in more detail can be found in \cite{fock_poisson_1998}.

There is a system of local holomorphic Darboux coordinates $\alpha_i, \beta_i$ on $\mathcal M_{g,n;\nu}$ defined in \cite{nekrasov_darboux_2013}, which we will call NRS coordinates.
They depend on a choice of pants decomposition and are Darboux coordinates on an appropriate open chart of the character variety.

\paragraph{$G=SL(2,\mathbb C)$ and $g=0$.} 
Now we restrict to $G=SL(2,\mathbb C)$ and to the Riemann sphere with $n=4$ punctures. The moduli space $\mathcal M_{0,4;\nu}$ is parametrized by four-tuples of elements of $SL(2,\mathbb C)$:
\begin{equation}
	\mathcal M_{0,4;\nu} = 
	\left\{ (g_1,g_2,g_3,g_4)\in G^4: 
	g_1g_2g_3g_4 = 1, \Tr g_i = \mu_i = 2\cos(2\pi\nu_i) \right\} / G .
\end{equation}
The group $G$ acts by simultaneous conjugation of all $g_i$. 

A useful set of functions on $\mathcal M_{0,4;\nu}$ is given by traces of monodromies around pairs of punctures:
\begin{equation}
	\begin{split}
		A &= \Tr(g_1g_2),\\
		B &= \Tr(g_2g_3),\\
		C &= \Tr(g_1g_3).
	\end{split}
\end{equation}
The smooth part of $\mathcal M_{0,4;\nu}$ is two complex dimensional. Thus, the functions $A$, $B$ and $C$ are not independent. They satisfy the Fricke relation
\begin{multline}
W_{0,4}(A,B,C) = ABC+A^2+B^2+C^2-4
-A(\mu_3\mu_4+\mu_1\mu_2)-B(\mu_1\mu_4+\mu_2\mu_3)-\\
-C(\mu_1\mu_3+\mu_2\mu_4)+\mu_1^2+\mu_2^2+\mu_3^2+\mu_4^2+\mu_1\mu_2\mu_3\mu_4=0 .
\label{eq:constraint04}
\end{multline}

This relation holds for any choice of $g_i\in SL(2,\mathbb C)$ such that $g_1g_2g_3g_4=e$ and $\Tr g_i = \mu_i$.

The functions $A, B, C$ can be expressed in terms of the NRS coordinates $\alpha,\beta$ on 
$\mathcal M_{0,4;\nu}$ \cite{nekrasov_darboux_2013}:
\begin{equation}
\label{eq:alphabetadefinition}
\begin{split}
    A&=e^\alpha+e^{-\alpha},\\
    B(A^2-4)+2(\mu_2\mu_3+\mu_1\mu_4)-A(\mu_1\mu_3+\mu_2\mu_4)
    &=\sqrt{c_{12}(A)c_{34}(A)}(e^\beta+e^{-\beta}),\\
    (2C + AB -\mu_1\mu_3-\mu_2\mu_4)(e^\alpha - e^{-\alpha}) 
    &=(e^\beta-e^{-\beta}) \sqrt{c_{12}(A)c_{34}(A)},
\end{split}
\end{equation}
where
\begin{equation}
    c_{ij}(A) = A^2+\mu_i^2+\mu_j^2-A\mu_i\mu_j-4 .
\end{equation}
The coordinates \(\alpha,\beta\) are multi-valued because they involve choices of logarithms and square roots. Once these branches are fixed, they form Darboux coordinates on the corresponding open chart.

\section{Opers and the semiclassical limit}

An $SL(2,\mathbb C)$-oper can be described, in the present rank-two setting, as a second-order differential operator $\mathcal D$ acting on the space $\Omega^{-1/2}$ of $(-\frac{1}{2})$-differentials:
\begin{equation}
\label{eq:oper}
	\mathcal D: \Omega^{-1/2} \to  \Omega^{3/2}.
\end{equation}
Equivalently, it is a projective connection. This differential-operator description is the form most useful for the WKB analysis below.

Locally, the $SL(2,\mathbb C)$-oper is a differential operator of the form 
$\mathcal D = -\partial_z^2+ \hbar^{-2} T(z)$. The definition \eqref{eq:oper} provides the transformation rule for $T(z)$ under a change of coordinates $z=z(w)$:
\begin{equation}
	(-\partial_w^2 + \hbar^{-2}\tilde T(w)) (z')^{-1/2} 
	= (z')^{3/2}(-\partial_z^2 + \hbar^{-2}T(z(w))) .
\end{equation}

Thus, the transformation rule is
\begin{equation}
	\tilde T(w) = (z')^2 T(z(w)) + \frac{\hbar^2}{2} \{z, w\}, 
\end{equation}
where
\begin{equation}
	\{z,w\} = \left(\frac{z''}{z'}\right)' - \frac{1}{2} \left(\frac{z''}{z'}\right)^2
\end{equation}
is the Schwarzian derivative. 

A second-order differential operator $\mathcal D=-\partial_z^2 + \hbar^{-2}T(z)$ is equivalent to a holomorphic connection of the form
\begin{equation}
	\nabla = \partial_z + 
	\begin{pmatrix}
		0 & -\hbar^{-2} T(z) \\
		-1 & 0
	\end{pmatrix}.
\end{equation}

In the $g=0$, $n=4$ case, the $SL(2,\mathbb C)$-oper is determined by a meromorphic function $T(z)$ on 
$\mathbb{CP}^1$ with at most second-order poles at the punctures. 
In an affine coordinate in which all punctures $z_i$ are finite, the absence of an additional pole at infinity imposes $T(z)=O(z^{-4})$ as $z\to\infty$
The generic function that satisfies these conditions has the form
\begin{equation}
	T(z) = \sum_{i=1}^n \left(\frac{\hbar^2 \Delta_i}{(z-z_i)^2}+\frac{\hbar^2 \epsilon_i}{z-z_i}\right).
\end{equation}
Here
\begin{equation}
	\Delta_i = \nu_i(\nu_i-1)
\end{equation}
is determined by the local monodromy exponent. 
The local exponents of the scalar equation are \(\nu_i\) and \(1-\nu_i\), so the corresponding monodromy trace is $2\cos(2\pi\nu_i)$, up to the usual ambiguity $\nu_i\sim 1-\nu_i\sim \nu_i+1$

The complex numbers $\epsilon_i$ are called the accessory parameters. For $T(z)$ to be regular at infinity they must satisfy 
\begin{equation}
	\label{eq:epsiloncond}
\begin{split}
    \sum_{i=1}^n \epsilon_i &= 0, \\
    \sum_{i=1}^n (\epsilon_i z_i +\Delta_i) &= 0,\\
    \sum_{i=1}^n (\epsilon_i z_i^2 + 2\Delta_i z_i) &= 0.
\end{split}
\end{equation}

The set of accessory parameters $\epsilon_i$ together with the positions of poles $z_i$ defines the $SL(2,\mathbb C)$-oper on $\Sigma$, which in turn determines a point of the character variety $\mathcal M_{0,n;\nu}$ through its monodromy representation. 
For fixed puncture positions and fixed local exponents \(\nu_i\), the three equations \eqref{eq:epsiloncond} leave \(n-3\) independent accessory parameters. If the puncture positions are also allowed to vary, then after quotienting by \(PGL(2,\mathbb C)\) the total number of parameters is
\[
(n-3)+(n-3)=2n-6.
\]
For \(n=4\), this gives a two-dimensional space, matching the dimension of \(\mathcal M_{0,4;\nu}\).

We keep the four punctures \(z_1,z_2,z_3,z_4\) finite and distinct.
The complex structure of the four-punctured sphere is encoded by their
cross-ratio
\[
q=\frac{(z_1-z_3)(z_2-z_4)}{(z_1-z_4)(z_2-z_3)} .
\]
One may use \(PGL_2(\mathbb C)\) to place three punctures at convenient finite values, but we will not put any puncture at infinity. For a complex structure, opers form a Lagrangian $n-3$-dimensional variety in $\mathcal M_{0,n;\nu}$.

The semiclassical limit considered later is the limit \(\hbar\to0\) in which the quadratic differential \(T(z)\) has a nontrivial limit. Equivalently, one keeps
\[
\mathfrak{m}_i=\hbar\nu_i,\qquad \tilde \epsilon_i=\hbar^2\epsilon_i
\]
fixed. In this limit
\[
T(z)=\sum_i\left(\frac{\mathfrak{m}_i^2+O(\hbar)}{(z-z_i)^2}
+\frac{\tilde\epsilon_i}{z-z_i}\right)
\]
becomes the quadratic differential defining the Seiberg--Witten curve \(p^2=T(z)\). Thus the WKB limit of the oper equation is the limit in which the oper data degenerate to the classical Seiberg--Witten geometry.

One can isolate the single Coulomb branch parameter $U\in \mathbb C$ by studying variations of $T(z)$ that keep the second-order poles fixed. There is only one meromorphic function one can write that has single-order poles at $z_i$ and has $O(z^{-4})$ asymptotic at infinity. One obtains
\begin{equation}
    T(z;U) = \frac{P(z)}{\prod_{i=1}^4(z-z_i)^2}+U\prod_{i=1}^4 \frac{1}{(z-z_i)},
    \label{eq:T_with_single_parameter}
\end{equation} 
where $\frac{P(z)}{\prod_{i=1}^4(z-z_i)^2}$ has pole structure fixed by $z_i$ and $\mathfrak{m}_i$.

\section{Exact WKB analysis and Stokes graphs}
\label{sec:WKB}

Throughout this section we take the semiclassical limit \(\hbar\to0\) with the quadratic differential \(T(z)\,\dd z^2\) fixed. Equivalently, \(\mathfrak{m}_i=\hbar\nu_i\) and \(\tilde\epsilon_i=\hbar^2\epsilon_i\) are held fixed. When \(\hbar\) is allowed to be complex, the Stokes graph depends only on the phase \(\vartheta=\arg\hbar\); if not stated otherwise we take \(\hbar>0\), so \(\vartheta=0\). In terms of the four-dimensional theory, the parameters $\mathfrak{m}_1,\mathfrak{m}_2,\mathfrak{m}_3,\mathfrak{m}_4$ determine the combinations $\mathfrak{m}_1\pm \mathfrak{m}_2$ and $\mathfrak{m}_3\pm \mathfrak{m}_4$ of the four fundamental hypermultiplet masses of the $SU(2)$ theory, while the cross-ratio $q$ is related to the gauge coupling by $q=e^{2\pi i\tau_{UV}}$.

\paragraph{Exact WKB method.} The goal of the WKB approximation is to obtain the monodromies of the Schrödinger equation in the complex plane
\begin{equation}
    \left(\frac{\dd^2}{\dd z^2}-\frac{1}{\hbar^2}T(z)\right)\psi(z)=0
    \label{eq:Schrodinger}
\end{equation}
The sign in \eqref{eq:Schrodinger} is chosen so that the spectral curve is \(p^2=T(z)\). Here $T(z)$ is a meromorphic quadratic differential coefficient on $\mathcal C$ with double poles at the punctures.

Let $\psi(z)=\exp R(z)$ and set $S(z)=\partial_z R(z)$. Then $S$ satisfies the Riccati equation
\begin{equation}
    S^2+\pdv{S}{z}=\frac{1}{\hbar^2}T(z)
    \label{eq:Sequation}
\end{equation}
We look for two formal solutions of the form
\begin{equation}
    S^{(\pm)}(z,\hbar)=\pm\frac{1}{\hbar}S_{-1}(z)+S_0(z)\pm \hbar S_1(z)+\hbar^2S_2(z)\pm\cdots
\end{equation}
Equivalently, \(S^{(+)}\) and \(S^{(-)}\) are exchanged by the deck transformation of the double cover \(\Sigma \to\mathcal C\). The coefficients are determined recursively. In particular,
\begin{equation}
\begin{split}
    S_{-1}^2 &= T,\\
    2S_{-1}S_j &= -\left(\sum_{k+l=j-1;\; k,l\ge0}S_kS_l+\pdv{S_{j-1}}{z}\right),\qquad j\ge0
\end{split}
\label{eq:recurrence}
\end{equation}
Here the sign choice of \(S_{-1}=p\) chooses a sheet of the spectral curve
\begin{equation}
    \Sigma: \qquad p^2=T(z).
\end{equation}

Define
\begin{equation}
    S_{\mathrm{odd}}=\frac{S^{(+)}-S^{(-)}}{2}=\frac{1}{\hbar}S_{-1}+\hbar S_1+\hbar^3S_3+\cdots,
\end{equation}
\begin{equation}
    S_{\mathrm{even}}=\frac{S^{(+)}+S^{(-)}}{2}=S_0+\hbar^2S_2+\hbar^4S_4+\cdots
\end{equation}
Substituting $S=S_{\mathrm{odd}}+S_{\mathrm{even}}$ into \eqref{eq:Sequation} and taking the part odd under the sheet involution gives
\begin{equation}
    S_{\mathrm{even}}=-\frac12\pdv{}{z}\log S_{\mathrm{odd}}.
\end{equation}
Thus the normalized formal WKB solutions can be written as
\begin{equation}
    \widehat\psi_{\pm,z_0}(z)=\frac{1}{\sqrt{S_{\mathrm{odd}}(z)}}
    \exp\left(\pm\int_{z_0}^z S_{\mathrm{odd}}(z,\hbar)\,\dd z\right).
    \label{eq:formal-wkb-solutions}
\end{equation}
The coefficients $S_j(z)$ are meromorphic functions on the spectral curve $\Sigma$; away from the singularities of the differential equation their possible poles occur at zeroes and poles of $T$. This follows recursively from \eqref{eq:recurrence}: differentiating and multiplying already constructed meromorphic functions cannot create new singularities, while division by $S_{-1}=p$ can introduce poles at the branch points of $\Sigma$.

The series \eqref{eq:formal-wkb-solutions} is formal and is usually divergent. Exact WKB analysis assigns actual local solutions by Borel transforming and Laplace summing the formal series in sectors where the relevant Borel sums exist. This is the sense in which the WKB solutions are "exact." We will use only the standard consequences of this theory; see \cite{voros_return_1983,kawai_takei_algebraic_2005,iwaki_nakanishi_exact_2014} for general references and \cite{gaiotto_wall-crossing_2011} for the relation to WKB triangulations in the Hitchin-system setting.

The Stokes curves are the trajectories along which the two exponential factors in \eqref{eq:formal-wkb-solutions} have equal phase. More invariantly, for phase $\vartheta=\arg\hbar$ they are the horizontal trajectories of the rotated quadratic differential $e^{-2i\vartheta}T(z)\dd z^2$:
\begin{equation}
    \Im\left(e^{-i\vartheta}\int_{u_i}^z \lambda\right)=0, \quad \lambda = p(z)\dd z
    \label{eq:Stokes_curves_definition}
\end{equation}
A simple zero $u_i$ of $T$ is a simple turning point, and three Stokes curves emanate from it in the $z$-plane.

A Stokes curve joining two zeroes is usually called a saddle connection or finite Stokes trajectory. Its existence is a real codimension-one condition
\begin{equation}
    \Im\left(e^{-i\vartheta}\int_{u_i}^{u_j}\lambda\right)=0 
\end{equation}
In the generic chambers considered below we assume that there are no saddle connections and no closed annular domains, so each separating Stokes curve starting at a turning point ends at a pole. At the exceptional phases the Stokes graph mutates; this is the exact-WKB version of the flip of a WKB triangulation, and the corresponding Voros symbols undergo cluster transformations \cite{iwaki_nakanishi_exact_2014,allegretti_voros_2019}.

Let $\{U_i\}$ be the Stokes regions, i.e. the connected components of the complement of the Stokes graph. In each region $U_i$ a Borel-summed pair of WKB solutions gives a local basis
\[
    \psi^i(z)=\begin{pmatrix}\psi_+^i(z)\\ \psi_-^i(z)\end{pmatrix}.
\]
If $U_i$ and $U_j$ are adjacent Stokes regions, analytic continuation across the common Stokes curve gives a triangular Stokes matrix:
\begin{equation}
    \begin{pmatrix}
        \psi_+^j(z)\\
        \psi_-^j(z)
    \end{pmatrix}
    =M^{ij}(z_0,\gamma)
    \begin{pmatrix}
        \psi_+^i(z)\\
        \psi_-^i(z)
    \end{pmatrix}
    \label{eq:stokes-change-of-basis}
\end{equation}
The matrix depends on the normalization point $z_0$, the path used for analytic continuation, and the branch choices for $p$ and the square root in \eqref{eq:formal-wkb-solutions}.

We now describe the local behavior of the Stokes curves. Near a pole $z_i$ we have
\begin{equation}
    e^{-i\vartheta}\lambda
    =\left(\frac{e^{-i\vartheta}\mathfrak{m}_i}{z-z_i}+O(1)\right)\dd z
\end{equation}
If $v$ is defined by $(e^{-i\vartheta}\lambda)(v)=1$, then
\begin{equation}
    v=\left(\frac{e^{i\vartheta}(z-z_i)}{\mathfrak{m}_i}+O((z-z_i)^2)\right)\pdv{}{z}.
\end{equation}
The leading flow is
\begin{equation}
    z(t)-z_i=C\exp\left(\frac{e^{i\vartheta}}{\mathfrak{m}_i}t\right)
\end{equation}
so the Stokes curves approach the pole as logarithmic spirals, unless $e^{i\vartheta}/\mathfrak{m}_i$ is purely imaginary, in which case the local trajectories are circles. We exclude this degenerate case in the generic analysis below. We label a pole by $+$ if the chosen orientation of the Stokes trajectories is inward and by $-$ if it is outward.

Near a simple zero $u_i$ of $T$,
\begin{equation}
    \lambda=\left(k_i(z-u_i)^{1/2}+O((z-u_i)^{3/2})\right)\dd z.
\end{equation}
Thus, after passing to the local coordinate $w=\int_{u_i}^z\lambda$, the Stokes curves are the three rays on which $e^{-i\vartheta}w$ is real. On one branch of the square root, three trajectories meet the turning point, with adjacent sectors alternating between dominant and subdominant WKB behavior.

The Stokes graph $\mathcal S$ is the graph whose vertices are the zeroes and poles of $T$ and whose edges are the separating Stokes curves connecting them. In the generic situation above, it records the combinatorial data needed for the Stokes matrices and hence for the monodromy computation.

To compute monodromy along closed curves we need the Stokes matrices. If the WKB basis is normalized at the turning point $u_k$ from which the crossed Stokes curve emanates, the local connection formulas take the form
\begin{equation}
    M^{ij}(u_k)=
    \begin{pmatrix}
        1&0\\
        \pm i&1
    \end{pmatrix},
    \qquad \Re(e^{-i\vartheta}\lambda)<0,
    \label{eq:transition_matrices_negative}
\end{equation}
or
\begin{equation}
    M^{ij}(u_k)=
    \begin{pmatrix}
        1&\pm i\\
        0&1
    \end{pmatrix},
    \qquad \Re(e^{-i\vartheta}\lambda)>0
    \label{eq:transition_matrices_positive}
\end{equation}
The sign is positive for counterclockwise crossing around the turning point and negative for clockwise crossing, with the branch and normalization conventions fixed once and for all. These are the standard simple-turning-point Stokes matrices; different normalizations conjugate them by diagonal matrices.

If the basis is instead normalized at a common point $z_0$, let $\gamma_k$ be the path from $z_0$ to $u_k$ used to transport the WKB basis and define
\begin{equation}
    D_k=\begin{pmatrix}
        \exp\left(\int_{\gamma_k}S_{\mathrm{odd}}\,\dd z\right)&0\\
        0&\exp\left(-\int_{\gamma_k}S_{\mathrm{odd}}\,\dd z\right)
    \end{pmatrix}
\end{equation}
Then
\begin{equation}
    M^{ij}(z_0)=D_kM^{ij}(u_k)D_k^{-1}
\end{equation}
For convenience set
\begin{equation}
    w_k=\exp\left(2\int_{\gamma_k}S_{\mathrm{odd}}(z,\hbar)\,\dd z\right)
\end{equation}
Then the transition matrices at the common base point are
\begin{equation}
    M^{ij}(z_0)=
    \begin{pmatrix}
        1&0\\
        \pm i w_k^{-1}&1
    \end{pmatrix},
    \qquad \Re(e^{-i\vartheta}\lambda)<0,
\end{equation}
\begin{equation}
    M^{ij}(z_0)=
    \begin{pmatrix}
        1&\pm i w_k\\
        0&1
    \end{pmatrix},
    \qquad \Re(e^{-i\vartheta}\lambda)>0
\end{equation}

Finally, when the path goes once around a regular singularity $z_i$, the local WKB basis also acquires the formal local monodromy
\begin{equation}
    B_i=\begin{pmatrix}n_i&0\\0&n_i^{-1}\end{pmatrix},
    \qquad n_i=e^{2\pi i\nu_i}
    \label{eq:local-formal-monodromy}
\end{equation}

\paragraph{Classification of Stokes graphs for $n=4$.} For $SL(2)$-opers on the four-punctured sphere, the limiting quadratic differential $T(z)\dd z^2$ has four double poles. Since the divisor of a quadratic differential on $\mathbb{CP}^1$ has degree $-4$, the total number of zeroes is $2n-4=4$. 
Assuming all zeroes are simple, each zero emits three Stokes curves, so the generic separating Stokes graph has $3(2n-4)=12$ Stokes edges.

Under the generic completeness assumption above, every separating trajectory ends at one of the four double poles. Thus $\mathcal S$ is a bipartite graph with four zero-vertices and four pole-vertices. The Euler characteristic gives $F=2-V+E=2-8+12=6$
Moreover, each face has even boundary length at least four. Since the average boundary length is $\frac{2E}{F}=\frac{24}{6}=4$ each face is necessarily a quadrangle.

Because every face of $\mathcal S$ is a quadrangle, we can construct a dual graph $\widetilde{\mathcal S}$ whose vertices are the poles $z_i$ and whose edges are the diagonals of the quadrilateral faces connecting pairs of poles. The graph $\widetilde{\mathcal S}$ is the WKB triangulation associated with the Stokes graph, in the sense used in \cite{gaiotto_wall-crossing_2011,iwaki_nakanishi_exact_2014}. The six Stokes graphs shown in Figure \ref{fig:StokesGraphs} are all possible combinatorial configurations of Stokes lines with $4$ poles and $4$ turning points, as shown in \cite{kawai_takei_algebraic_2005}.

\begin{figure}[h]
\centering
\includegraphics[width=0.4\textwidth]{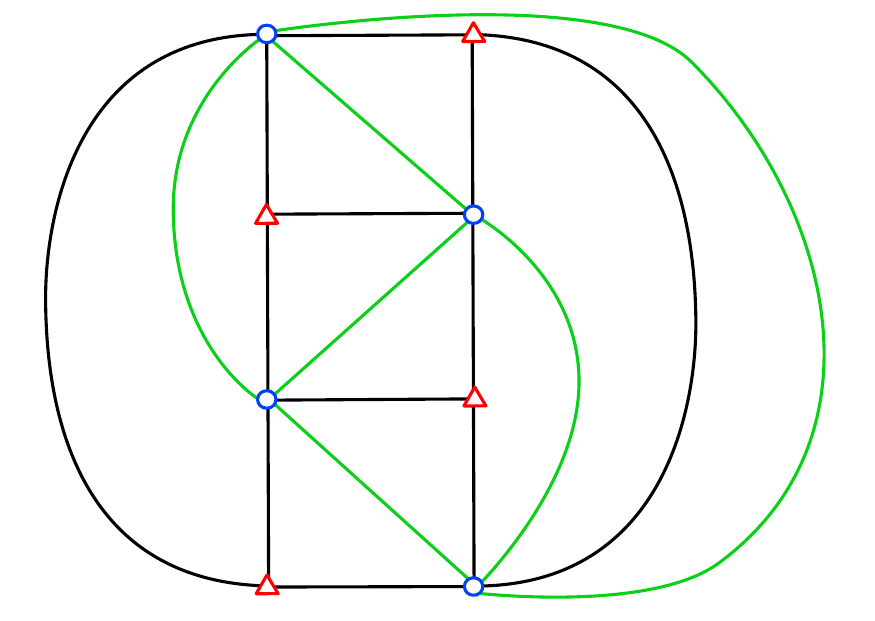}
\caption{Construction of $\widetilde {\mathcal S}$. The graph $\widetilde{\mathcal S}$ is shown in green and the Stokes graph $\mathcal S$ in black.}
\label{fig:StokesDualGraph}
\end{figure}

\begin{figure}[h]
\centering
    \subfloat[Type 1]{{\includegraphics[height=0.3\textwidth]{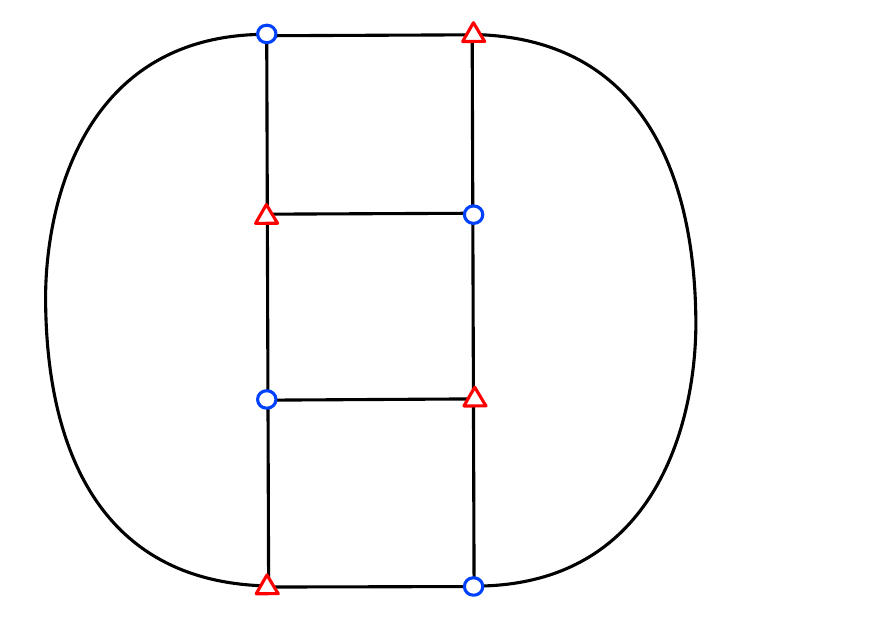} }}%
    \hfill
     \subfloat[Type 2]{{\includegraphics[height=0.3\textwidth]{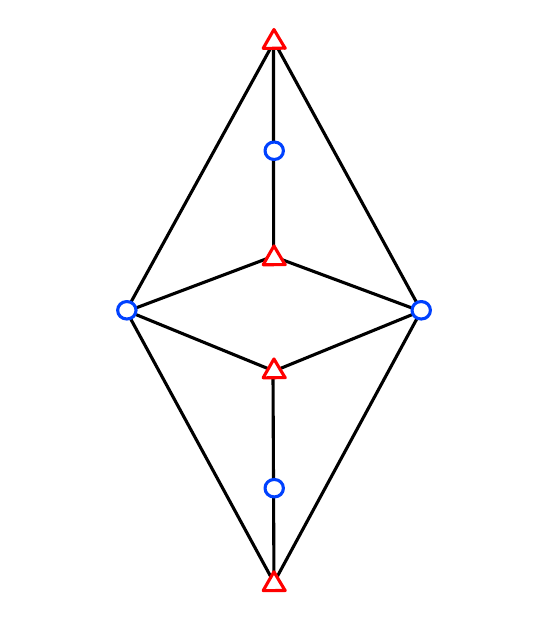} }}%
    \hfill
     \subfloat[Type 3]{{\includegraphics[height=0.3\textwidth]{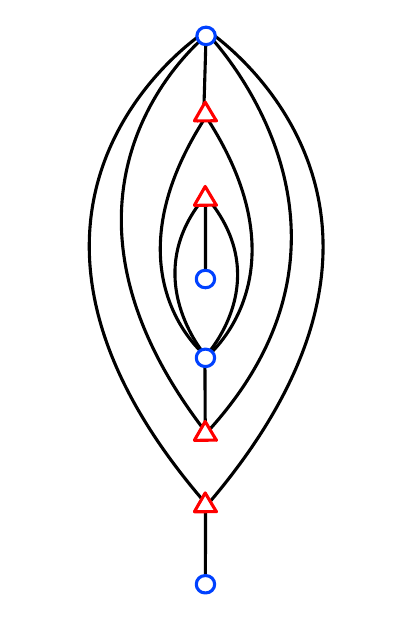} }}%
    \hfill
     \subfloat[Type 4]{{\includegraphics[height=0.3\textwidth]{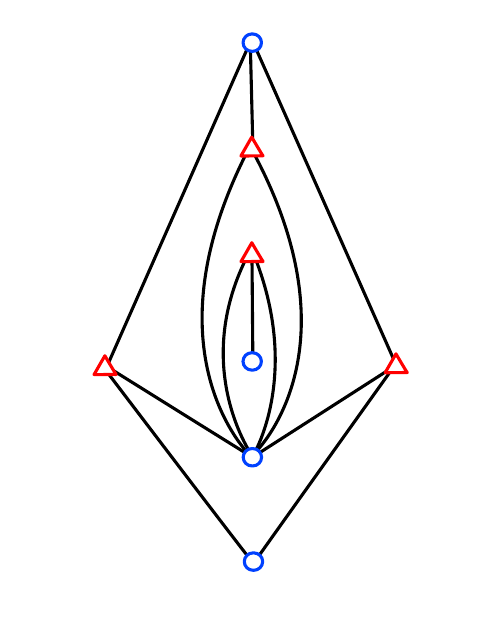} }}%
    \hfill
     \subfloat[Type 5]{{\includegraphics[height=0.3\textwidth]{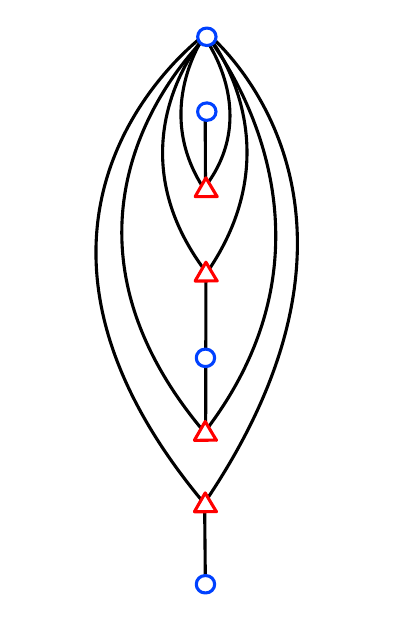} }}%
    \hfill
     \subfloat[Type 6]{{\includegraphics[height=0.3\textwidth]{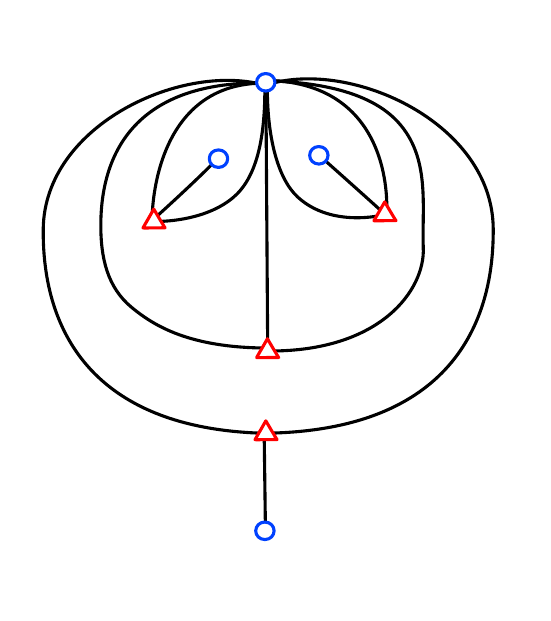} }}%
\caption{Six labeled types of Stokes graph for the four-punctured sphere. Calculations for each type are given in the Appendix.}
\label{fig:StokesGraphs}
\end{figure}

\paragraph{Calculation of monodromies.} Our goal is to calculate the trace functions
\[
    A=\Tr(g_1g_2),\qquad B=\Tr(g_2g_3),\qquad C=\Tr(g_1g_3),
\]
where $g_i$ is the monodromy around the pole $z_i$. The algorithm is the same in each Stokes chamber:
\begin{enumerate}
    \item Introduce cuts connecting pairs of zeroes, so that a single branch of $p=\sqrt{T}$ is chosen on the cut surface.
    \item Label poles by $+$ and $-$ according to the sign convention above. In each quadrilateral face not crossed by a cut, there should be one inward and one outward pole.
    \item Choose a base point $z_0$.
    \item Choose paths $\gamma_i$ from $z_0$ to the turning points $u_i$; these determine the variables $w_i$.
    \item Draw the loops representing the monodromies $g_i$.
    \item Write the Stokes transition matrices $M_{ij}(z_0)$ in terms of $w_i$ and $n_i$.
    \item Multiply the Stokes matrices and the corresponding local monodromy matrix $B_i=\mathrm{diag}(n_i,n_i^{-1})$ in the order dictated by the loop.
\end{enumerate}

We demonstrate the calculation for the Stokes graph of Type 1. The pictures and formulas for the remaining Stokes graphs can be placed in Appendix \ref{sec:Monodromies}.

\begin{figure}[h]
\centering
\includegraphics[width=0.5\textwidth]{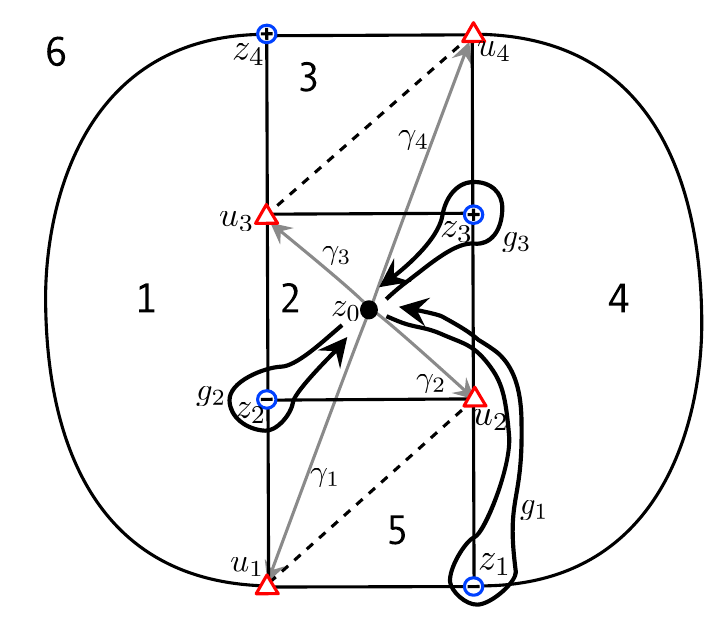}
\caption{Stokes graph of Type 1. Stokes regions are numbered $1$--$6$. Blue circles are poles, red triangles are zeroes, dashed lines are cuts, thin lines are Stokes curves, thick arrows are cycles $g_i$, and straight grey lines are the paths $\gamma_i$.}
\label{fig:StokesGraph1Mon}
\end{figure}

For the Type 1 graph in Figure \ref{fig:StokesGraph1Mon}, the monodromies $g_1$, $g_2$ and $g_3$ are
\begin{multline}
    g_1=M_{24}M_{45}M_{56}M_{64}M_{42}B_1\\
    =\begin{pmatrix}1&-iw_2\\0&1\end{pmatrix}
    \begin{pmatrix}1&0\\-i/w_2&1\end{pmatrix}
    \begin{pmatrix}1&0\\-iw_1/w_2^2&1\end{pmatrix}
    \begin{pmatrix}1&0\\-i/(w_4n_3^2n_1^2)&1\end{pmatrix}
    \begin{pmatrix}1&i n_1^2w_2\\0&1\end{pmatrix}
    \begin{pmatrix}n_1&0\\0&n_1^{-1}\end{pmatrix},
\end{multline}
\begin{equation}
    g_2=M_{21}M_{15}M_{52}B_2
    =\begin{pmatrix}1&0\\-i/w_3&1\end{pmatrix}
    \begin{pmatrix}1&0\\-i/(w_1n_2^2)&1\end{pmatrix}
    \begin{pmatrix}1&0\\-i/(w_2n_2^2)&1\end{pmatrix}
    \begin{pmatrix}n_2&0\\0&n_2^{-1}\end{pmatrix}.
    \label{eq:g2-corrected}
\end{equation}
\begin{equation}
    g_3=M_{24}M_{43}M_{32}B_3
    =\begin{pmatrix}1&-iw_2\\0&1\end{pmatrix}
    \begin{pmatrix}1&-iw_4n_3^2\\0&1\end{pmatrix}
    \begin{pmatrix}1&-iw_3n_3^2\\0&1\end{pmatrix}
    \begin{pmatrix}n_3&0\\0&n_3^{-1}\end{pmatrix}.
\end{equation}

With $w_{12}=w_2/w_1$ and $w_{13}=w_3/w_1$, the trace functions are
\begin{equation}
\label{eq:ABCformulas}
\begin{split}
    A&=\Tr(g_1g_2)
    =-\frac{{n_2} {w_{12}}}{{n_1} {w_{13}}}-\frac{{n_1}{n_2}}{{w_{12}}}
    -\frac{{w_{12}}}{{n_1}{n_2}}-\frac{{n_1}{n_2}}{{w_{13}}}
    -\frac{{n_2}^2 {n_4} {w_{12}}}{{n_3}{w_{13}}^2}
    -\frac{{n_2}^2 {n_4}}{{n_3}{w_{13}}}
    -\frac{{n_4}{w_{12}}}{{n_3}{w_{13}}},\\
    B&=\Tr(g_2g_3)
    =-\frac{w_{13}w_{12}}{n_1n_2^2n_4}-\frac{w_{12}}{n_2n_3}
    -\frac{w_{12}}{n_1n_4}-\frac{n_2w_{12}}{n_3w_{13}}
    -\frac{n_3w_{13}}{n_2}-\frac{w_{13}}{n_1n_2^2n_4}
    -\frac{n_3w_{13}}{n_2w_{12}},\\
    C&=\Tr(g_1g_3)
    =-\frac{n_1n_3w_{13}}{w_{12}}-\frac{n_1n_3w_{13}}{w_{12}^2}
    -\frac{n_1n_3}{w_{12}}-\frac{w_{13}}{n_2n_4}
    -\frac{w_{13}}{n_2n_4w_{12}}-\frac{n_2n_4}{w_{12}}
    -\frac{n_2n_4}{w_{13}}.
\end{split}
\end{equation}
These formulas should be understood as the trace expressions in this chosen Stokes chamber and with the chosen branch/cut conventions. Across a saddle-connection wall, the variables \(w_{ij}\) transform by the corresponding Stokes/cluster transformation, so the tropical leading term of \(A,B,C\) may change even though the monodromy representation itself is analytically continued.

The monodromy around a contour enclosing all poles and cuts is trivial, which gives the relation
\begin{equation}
    n_1n_2n_3n_4\,w_1w_4w_2^{-1}w_3^{-1}=1.
\end{equation}
Thus any trace of monodromy around a closed curve can be expressed in terms of $w_{12}$, $w_{13}$ and $n_1,n_2,n_3,n_4$. Each of these variables is the exponential of an integral of the one-form
\begin{equation}
    \kappa=S_{\mathrm{odd}}(z,\hbar)\,\dd z
\end{equation}
around a path or cycle on the spectral curve $\Sigma_U$. In this sense, the exact WKB machinery maps Wilson loops of the Schrödinger connection on $\mathcal C$ to a linear combination of periods of the differential $\lambda=p(z)\dd z$ on the spectral curve $\Sigma$ that is a double cover of $\mathcal C$. In the limit $\hbar\to0$, these traces are controlled to leading order by one of the periods of the Seiberg-Witten differential $\lambda$.
\section{Tropical asymptotics and the Seiberg--Witten period map}
\label{sec:tropical}

We now state the chamberwise semiclassical result used below. Fix a Stokes
chamber and a ray
\[
    \hbar=\rho e^{i\vartheta},\qquad \rho\to0^+,
\]
and let
\[
    \Sigma_U:\quad p^2=T(z;U),\qquad \lambda=p\,dz
\]
be the limiting Seiberg--Witten curve and differential. The small loops around
the punctures give the flavor periods
\[
    \oint_{\delta_i}\lambda=2\pi i\,\mathfrak m_i,
\]
up to the sign determined by the sheet of the double cover.

The growth of a period \(\Pi\) along the chosen ray is governed by
\[
    \operatorname{val}_\vartheta(\Pi)
    :=
    \Re(e^{-i\vartheta}\Pi),
\]
since
\[
    \Re(\Pi/\hbar)=\rho^{-1}\operatorname{val}_\vartheta(\Pi).
\]
Thus, in every finite sum of WKB exponentials, the leading term is the one with
largest \(\vartheta\)-valuation.

\begin{definition}[Integral-period and non-integral-period chambers]
Fix the branch cuts, a branch of \(p=\sqrt T\), the paths defining the Voros
exponentials, and the branches of the logarithms and square roots entering the
NRS coordinate chart. A tropical subchamber is a connected region where the
dominant exponential terms in the relevant trace formulae and NRS relations are
fixed.

A tropical subchamber is called an integral-period chamber if the dominant terms
are unique, their leading coefficients do not cancel, and the leading NRS
logarithms are genuine Seiberg--Witten periods:
\[
    \alpha_{\mathrm{NRS}}
    =
    \frac{1}{\hbar}\int_{\gamma_\alpha}\lambda+O(1),
    \qquad
    \beta_{\mathrm{NRS}}
    =
    \frac{1}{\hbar}\int_{\gamma_\beta}\lambda+O(1),
\]
where \(\gamma_\alpha,\gamma_\beta\in H_1(\Sigma;\mathbb Z)\) form a primitive
symplectic pair,
\[
    \langle\gamma_\alpha,\gamma_\beta\rangle=\pm1.
\]
Equivalently, after quotienting by flavor-period shifts, the integer linear
system obtained by tropicalizing the NRS relations is non-degenerate and
unimodular.

A tropical subchamber is called non-integral-period if the dominant terms are
well defined and non-cancelling, but the resulting leading NRS logarithms do not
define such an integral symplectic period pair. This includes rational
combinations of periods, non-primitive cycles, degenerate systems, and pairs of
cycles whose intersection is not \(\pm1\). Walls where dominance is not unique,
or where leading coefficients cancel, are boundaries of the tropical
decomposition rather than chambers.
\end{definition}

\begin{proposition}[Tropical NRS leading system]
\label{prop:tropical-leading-system}
Fix a generic Stokes chamber, branch choices, and a ray of \(\hbar\). In a
tropical subchamber where all relevant trace and NRS relations have unique
dominant monomials and no leading cancellation, the leading logarithms
\[
    \alpha_{\mathrm{NRS}}=\frac{\alpha_0}{\hbar}+O(1),
    \qquad
    \beta_{\mathrm{NRS}}=\frac{\beta_0}{\hbar}+O(1)
\]
are determined by an integer affine linear system whose coefficients are read
off from the dominant monomials. The inhomogeneous terms are periods of
\(\lambda\), including possible flavor periods. If this system is unimodular
modulo flavor-period shifts and the resulting cycles form a primitive
symplectic pair, then the subchamber is an integral-period chamber. Otherwise
the same procedure gives a tropical asymptotic, but not an integral
Seiberg--Witten electric-magnetic basis in the chosen NRS chart.
\end{proposition}

\begin{proof}
By the WKB expansion, every regularized Voros integral appearing in the
monodromy computation has leading term
\[
    \int_\eta S_{\mathrm{odd}}(z,\hbar)\,dz
    =
    \frac{1}{\hbar}\int_\eta\lambda+O(1).
\]
The trace formulae obtained from Stokes matrices are finite Laurent sums in
Voros variables and in the local monodromy variables
\[
    n_i=e^{2\pi i\nu_i}
       =
    \exp\left(\frac{2\pi i\,\mathfrak m_i}{\hbar}\right).
\]
Therefore every monomial appearing in a trace coordinate has leading exponent
given by a period of \(\lambda\), including possible flavor-period contributions
from loops around punctures.

Inside a tropical subchamber, suppose a unique monomial dominates a trace
coordinate. Then its logarithm is the corresponding period divided by \(\hbar\),
up to \(O(1)\). Applying this to
\[
    A=\Tr(g_1g_2),\qquad
    B=\Tr(g_2g_3),\qquad
    C=\Tr(g_1g_3),
\]
we obtain
\[
    \log A=\frac{\Pi_A}{\hbar}+O(1),\qquad
    \log B=\frac{\Pi_B}{\hbar}+O(1),\qquad
    \log C=\frac{\Pi_C}{\hbar}+O(1),
\]
where \(\Pi_A,\Pi_B,\Pi_C\) are periods of \(\lambda\), with possible flavor
shifts.

Substituting these asymptotics into the NRS coordinate relations gives a finite
collection of exponential terms whose leading exponents are integer-linear
functions of the unknowns \(\alpha_0\) and \(\beta_0\), plus known period data.
Dominance of the selected terms in the NRS equations gives equalities of these
leading exponents, hence an integer affine linear system. The final assertion
is exactly the unimodularity and intersection-pairing condition stated above.
\end{proof}

\paragraph{When NRS does not reproduce SW}
Note that since all of the pairwise monodromies $A$, $B$, $C$ have at least two terms that are inverse of one another. This implies that in every chamber $|A|, |B|, |C|\to \infty$ as $\hbar\to 0$. Therefore, the leading asymptotic of $\alpha_{NRS}$ is fixed by the dominant period of $A$:
\begin{equation}
    \alpha_{NRS} \sim \log A \sim \frac{\Pi_A}{\hbar} 
\end{equation} 

As far as the second and third equations \eqref{eq:alphabetadefinition} defining $\alpha_{NRS}$ and $\beta_{NRS}$ go, the situation is more complicated. First of all, the dominating term under the square root need not be divisible by 2 as an element of $H^1(\Sigma, \mathbb{Z})$. Therefore, the square root need not correspond to a cycle of the SW curve. This is the first case when we are not in an integral-period chamber.

The second non-integral scenario is when, after all cancellations on the left side and division by the square root, the cycle $\gamma_{\beta}$ for the leading term doesn't have intersection pairing $\langle \gamma_{\alpha} , \gamma_\beta\rangle = \pm 1$. In that case, the symplectic form $\Omega = \dd \alpha_{NRS} \wedge \dd \beta_{NRS}$ doesn't reproduce $\dd a_{SW}\wedge \dd b_{SW}$. 

We interpret this as a coordinate-dependent limitation of the chosen pants
decomposition and NRS chart. It remains an open question whether, after
changing the pants decomposition or Darboux chart, every such chamber admits
an integral-period description.

\paragraph{Tropical walls}
The walls are where two competing monomials have equal \(\vartheta\)-valuation.
Equivalently, for some period difference \(\int_\gamma\lambda\),
\[
    \Re\left(\frac{1}{\hbar}\int_\gamma\lambda\right)=0.
\]
These include the usual Stokes walls where the Stokes graph changes, but also
walls where the Stokes graph is fixed while the dominant term in a trace or NRS
relation changes.

\paragraph{Consequence for the IR coupling.}
Inside a fixed integral-period chamber, assume that
\(\gamma_\alpha,\gamma_\beta\) are locally constant as \(U\) varies. Then
\[
    \frac{\partial\beta_{\mathrm{NRS}}}
         {\partial\alpha_{\mathrm{NRS}}}
    =
    \frac{
        \partial_U\int_{\gamma_\beta}\lambda
    }{
        \partial_U\int_{\gamma_\alpha}\lambda
    }
    +O(\hbar).
\]
For a rank-one family
\[
    T(z;U)
    =
    \frac{P(z)}{\prod_{i=1}^4(z-z_i)^2}
    +
    U\prod_{i=1}^4\frac{1}{z-z_i},
\]
the \(U\)-derivative of a period is, up to an overall constant, a period of the
holomorphic differential on the elliptic curve. Hence the chamberwise IR
coupling is
\[
    \tau_{\mathrm{IR}}
    =
    \frac{
        \int_{\gamma_\beta}
        dz/\sqrt{\prod_{j=1}^4(z-u_j)}
    }{
        \int_{\gamma_\alpha}
        dz/\sqrt{\prod_{j=1}^4(z-u_j)}
    }.
\]

\paragraph{A sample integral-period chamber.}
Consider the quadratic differential
\[
\begin{aligned}
T(z)=&
\frac{1}{(z-1-i)^2}
+\frac{1}{(z-1+i)^2}
+\frac{1}{(z+1-i)^2}
+\frac{1}{(z+1+2i)^2}
\\
&+\frac{0.5}{z-1-i}
+\frac{0.75+1.25i}{z+1-i}
+\frac{-1.25-1.25i}{z-1+i}.
\end{aligned}
\]
This corresponds to
\(q=11/26+3i/26\) and
\(\mathfrak m_1=\mathfrak m_2=\mathfrak m_3=\mathfrak m_4=1\). For
\(\arg\hbar=0\), the Stokes graph in Figure~\ref{fig:example_plot} is of
Type~1.

\begin{figure}[h]
\centering
\includegraphics[width=0.5\textwidth]{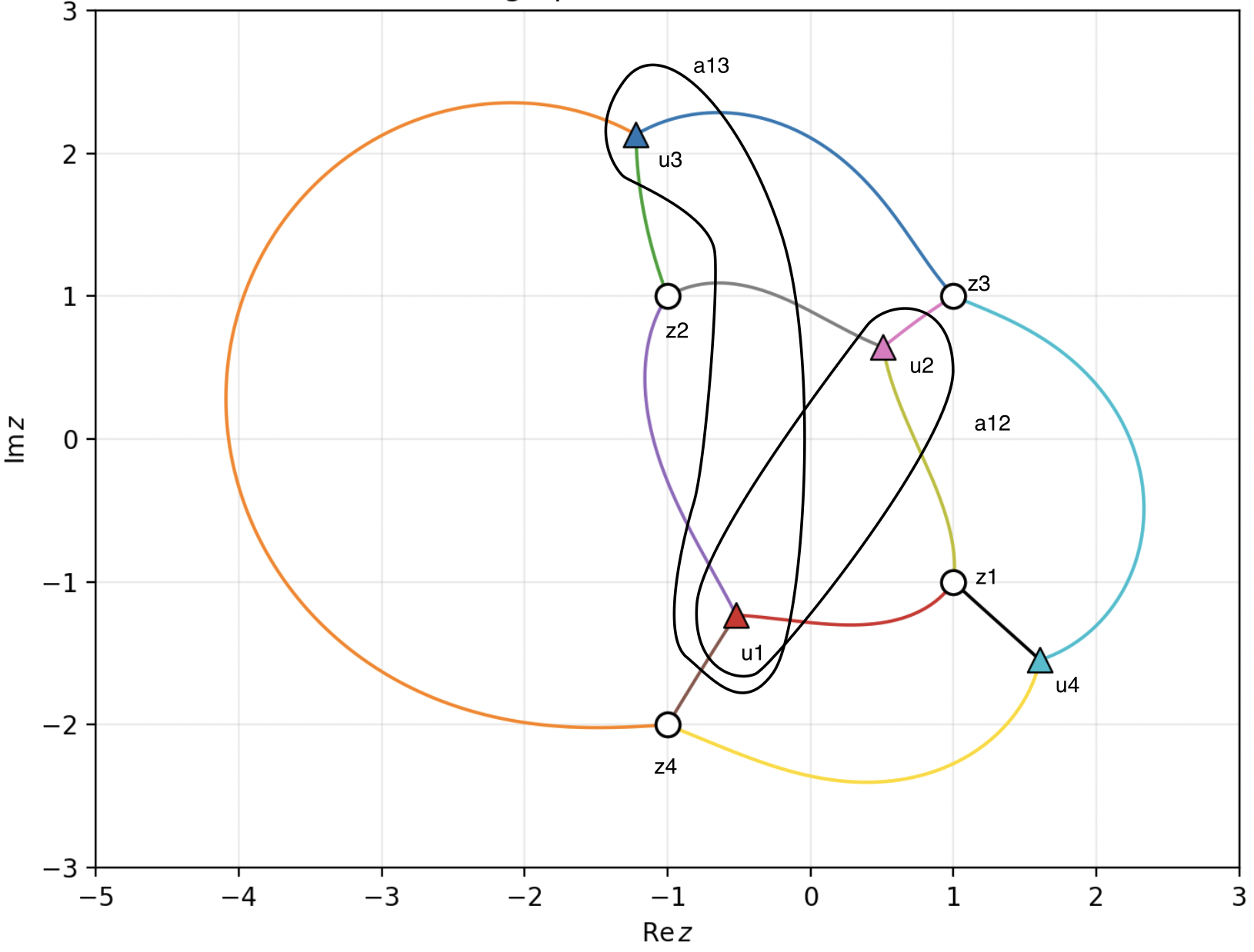}
\caption{Stokes graph for the example, with the cycles used to compute
\(a_{12}\) and \(a_{13}\) marked.}
\label{fig:example_plot}
\end{figure}

Using the Type~1 monodromy formulae, the dominant terms in this chamber select
specific periods for \(A,B,C\). After the relabeling
\((A,B,C)\mapsto(C,A,B)\), the NRS relations give
\[
    \alpha_{\mathrm{NRS}}
    =
    \frac{1}{\hbar}\int_{\gamma_\alpha}\lambda+O(1),
    \qquad
    \beta_{\mathrm{NRS}}
    =
    \frac{1}{\hbar}\int_{\gamma_\beta}\lambda+O(1),
\]
where the cycles are shown in Figure~\ref{fig:example_plot_AB}. Since these
cycles form a primitive integral pair, this is an integral-period chamber.

\begin{figure}[h]
\centering
\includegraphics[width=0.5\textwidth]{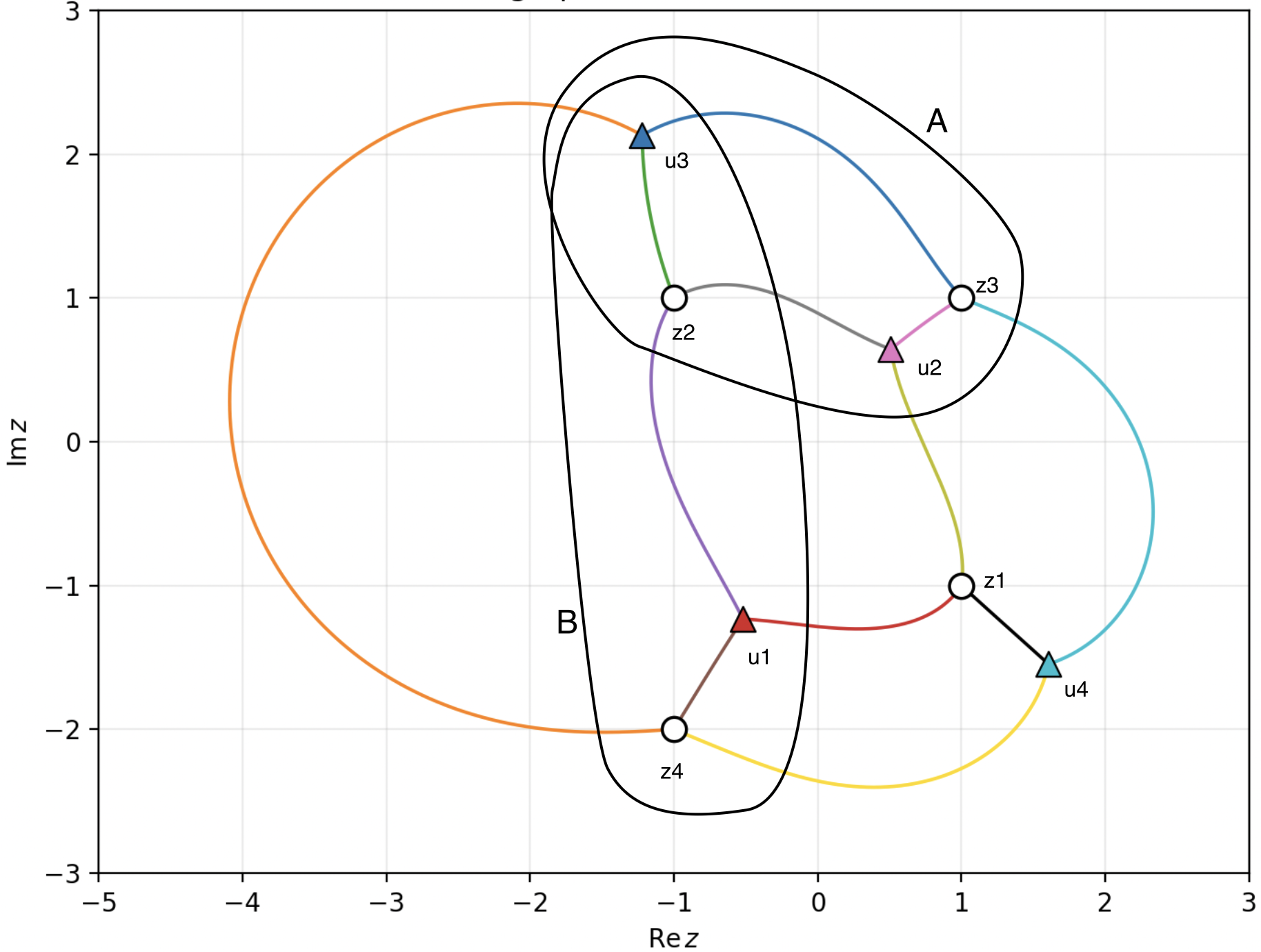}
\caption{Cycles corresponding to the leading asymptotics of \(A\) and \(B\),
equivalently of \(\alpha_{\mathrm{NRS}}\) and \(\beta_{\mathrm{NRS}}\).}
\label{fig:example_plot_AB}
\end{figure}

\section{Weak coupling limit}

It is difficult in general to decide whether a given tropical chamber is
integral-period: one must know whether the dominant terms in the NRS relations
select cycles with intersection pairing \(\pm1\). The weak-coupling limit is a
useful exception, because the relevant monodromy can be read directly from the
plumbing degeneration.

Normalize the punctures as
\[
    (z_1,z_2,z_3,z_4)=(0,q,1,\infty),
    \qquad q\to0.
\]
Then the four-punctured sphere degenerates into two three-punctured spheres
joined by a long neck. The punctures \(0\) and \(q\) lie on the inner sphere,
while \(1\) and \(\infty\) lie on the outer sphere. The neck carries the
internal residue \(a\), i.e. the electric Seiberg--Witten period.

We use the parametrization
\begin{equation}
\begin{aligned}
T(z;q)=&
\frac{\mathfrak m_1^2}{z^2}
+\frac{\mathfrak m_2^2}{(z-q)^2}
+\frac{\mathfrak m_3^2}{(z-1)^2}
+\frac{\mathfrak m_4^2-\mathfrak m_1^2-\mathfrak m_2^2-\mathfrak m_3^2}
     {z(z-1)} \\
&+
\frac{q(q-1)U}{z(z-q)(z-1)} .
\end{aligned}
\label{eq:weak-coupling-T}
\end{equation}
Keeping the internal residue fixed requires
\begin{equation}
    qU\longrightarrow a^2-\mathfrak m_1^2-\mathfrak m_2^2 .
    \label{eq:qU-a-relation}
\end{equation}
Thus \(U\) diverges like \(q^{-1}\) in the weak-coupling limit.

On the inner sphere, set \(z=qx\). Then
\[
    T(qx;q)q^2dx^2
    \longrightarrow
    \phi_{\mathrm{in}}(x)\,dx^2,
\]
where
\begin{equation}
\phi_{\mathrm{in}}(x)
=
\frac{\mathfrak m_1^2}{x^2}
+
\frac{\mathfrak m_2^2}{(x-1)^2}
+
\frac{a^2-\mathfrak m_1^2-\mathfrak m_2^2}{x(x-1)} .
\label{eq:phi-inner}
\end{equation}
This is the three-punctured differential with residues
\((\mathfrak m_1,\mathfrak m_2,a)\) at \(x=0,1,\infty\).

On the outer sphere, keeping \(z\) fixed gives
\[
    T(z;q)\,dz^2
    \longrightarrow
    \phi_{\mathrm{out}}(z)\,dz^2,
\]
where
\begin{equation}
\phi_{\mathrm{out}}(z)
=
\frac{a^2}{z^2}
+
\frac{\mathfrak m_3^2}{(z-1)^2}
+
\frac{\mathfrak m_4^2-a^2-\mathfrak m_3^2}{z(z-1)} .
\label{eq:phi-outer}
\end{equation}
This is the three-punctured differential with residues
\((a,\mathfrak m_3,\mathfrak m_4)\) at \(z=0,1,\infty\).

The two components are glued through the annulus
\[
    q\ll |z|\ll1.
\]
In this neck region,
\[
    T(z;q)\,dz^2
    =
    a^2\frac{dz^2}{z^2}
    +
    \text{subleading terms},
\]
so the Seiberg--Witten differential satisfies
\[
    \lambda\sim a\frac{dz}{z}.
\]

\paragraph{Quasiclassical behavior.}
The trace around the neck is, to leading order, the monodromy around the third
puncture of the inner three-punctured sphere:
\begin{equation}
    A\sim 2\cos\left(\frac{2\pi a}{\hbar}\right)+O(q).
\end{equation}
Equivalently, the dominant exponential in \(A\) has leading logarithm
\[
    \frac{1}{\hbar}\oint_A\lambda
    =
    \frac{2\pi i\,a}{\hbar}.
\]
Thus the tropical term selected by \(A\) is the neck cycle or its inverse, whichever dominates in a particular chamber.

The dual period of the Seiberg-Witten curve behaves as
\[
    w_{13}
    \sim
    \exp\left(\frac{a\log q}{\hbar}\right)+O(1).
\]
For sufficiently small \(q\), either \(w_{13}\) or \(w_{13}^{-1}\) dominates
the relevant tropical sums, depending on the ray of \(\hbar\). Note that the expressions for $A$ in the Appendix have either only non-negative or non-positive powers of $w_{13}$. The only exception is Type 6 of the Stokes graph, for which there is no way to split it into inner and outer regions with poles $z_1, z_2$ and zeroes $u_1, u_2$ ending up in the inner region, and the rest being in the outer region. Such splitting would include a Stokes line going into the other region, winding a pole there, and coming back, which is prohibited in the $q\to0$ limit. On the other hand, the expressions for $B$ in types 1-5 have terms with the opposite dependence on $w_{13}$, and the power of such terms is always $\pm 1$. 

It is helpful to switch $A$ and $B$ at this stage to obtain an explicitly integral-period chamber. This also involves exchanging $\mu_1$ and $\mu_3$ in the formulas. Using the modified definition 
\begin{equation}
\label{eq:alphabetadefinitionmod}
\begin{split}
    B&=e^{\alpha_{NRS}}+e^{-\alpha{NRS}},\\
    A(B^2-4)+2(\mu_1\mu_3+\mu_2\mu_4)-A(\mu_1\mu_2+\mu_3\mu_4)
    &=\sqrt{c_{23}(B)c_{14}(B)}(e^{\beta_{NRS}}+e^{-\beta_{NRS}}),
\end{split}
\end{equation}
we observe that in the leading order in $\log q$, these equations simplify to
\begin{equation}
    B=e^{\alpha_{NRS}}+e^{-\alpha_{NRS}}, \quad A=(e^{\beta_{NRS}}+e^{-\beta_{NRS}}),
\end{equation}
therefore explicitly giving the SW coordinates
\begin{equation}
    \alpha_{NRS} \sim \frac{1}{\hbar}\Pi_B, \quad \beta_{NRS}\sim \frac{1}{\hbar} \Pi_A
\end{equation}
with intersection pairing $\pm 1$. Therefore, in the plumbing-compatible NRS chart and away from the chamber walls, in the weak-coupling limit NRS coordinates recover the SW electromagnetic dual pair, up to flavor period shifts. 

\section{Conclusion}
In this paper we studied the semiclassical limit of the oper
Lagrangian inside the character variety of flat \(SL_2(\mathbb C)\)
connections on the four-punctured sphere. Using exact WKB monodromy
formulae, we expressed the trace coordinates \(A,B,C\) in terms of
Voros-type exponentials and then translated their asymptotics into the
Nekrasov--Rosly--Shatashvili Darboux coordinates \((\alpha,\beta)\).

The main outcome is a chamberwise tropical asymptotic calculation, not a global
coordinate-independent recovery theorem. The exact-WKB trace formulae give
period-valued Laurent sums, and the NRS relations tropicalize to an integer
affine linear system for the leading logarithms of \((\alpha,\beta)\). In
verified integral-period chambers, including the weak-coupling chamber above,
the resulting leading asymptotics of the NRS coordinates are given by a
primitive integral symplectic pair of periods of the Seiberg--Witten
differential
\[
    \lambda=p\,dz,
    \qquad
    p^2=T(z).
\]
In non-integral-period chambers, the WKB tropical asymptotics can still exist,
but the chosen NRS coordinate chart does not produce an integral
electric-magnetic period basis. Those chambers should be interpreted as
coordinate-dependent limitations of this chart rather than as a contradiction
of the Seiberg--Witten limit of the oper variety.

One case where we find it possible to always provide an integral-chamber description is the weak-coupling limit. In that limit, we can always adapt the NRS coordinates in a way that in the limit $q\to 0$, $\hbar \to 0$ we obtain SW periods as the leading asymptotic of $\alpha_{NRS}$ and $\beta_{NRS}$, as long as one stays within the tropical chamber.

The relation with the GMN construction is close but not identical. GMN use WKB
triangulations to construct Fock--Goncharov coordinates and study their
wall-crossing. Here the starting point is instead the exact WKB analysis of
\(SL_2\)-opers and the resulting monodromy traces, which are then expressed in
NRS coordinates adapted to a pants decomposition. This gives a complementary
route from oper monodromy data to Seiberg--Witten period geometry, but only
after identifying which tropical chambers are integral-period chambers.

Several extensions remain open. First, one should give an intrinsic criterion
for deciding whether a chamber is integral-period without explicitly tracing
Stokes lines and comparing dominant exponentials. Second, one should describe
how non-integral-period chambers are related to integral-period chambers by
cluster, Fenchel--Nielsen, or other changes of Darboux coordinates. Third, one
should compare the obtained piecewise-linear asymptotics more directly with
wall-crossing and related integrable-system structures. These questions are
necessary for turning the chamberwise calculation into a coordinate-independent
description of the semiclassical oper variety.

\appendix

\section{Monodromies for different Stokes graphs}

In this appendix, the WKB expressions for the monodromies $g_1$, $g_2$ and $g_3$ are given and the coordinates $A,B, C$ are calculated. There are two consistency checks. First, all of the trace formulae satisfy the Fricke relation. Second, the trace formulae depend nontrivially on the torus-cycle variables $w_{12}$ and $w_{13}$. This check is weaker than the integral-period condition: even when the trace formulae involve torus cycles, a particular tropical chamber may still be non-integral-period because the induced NRS asymptotics can be rational, non-primitive, degenerate, or dominated by flavor-period shifts after solving the NRS relations.

\label{sec:Monodromies}
\subsection{Type 1}

\begin{figure}[h]
\centering
\includegraphics[height=0.5\textwidth]{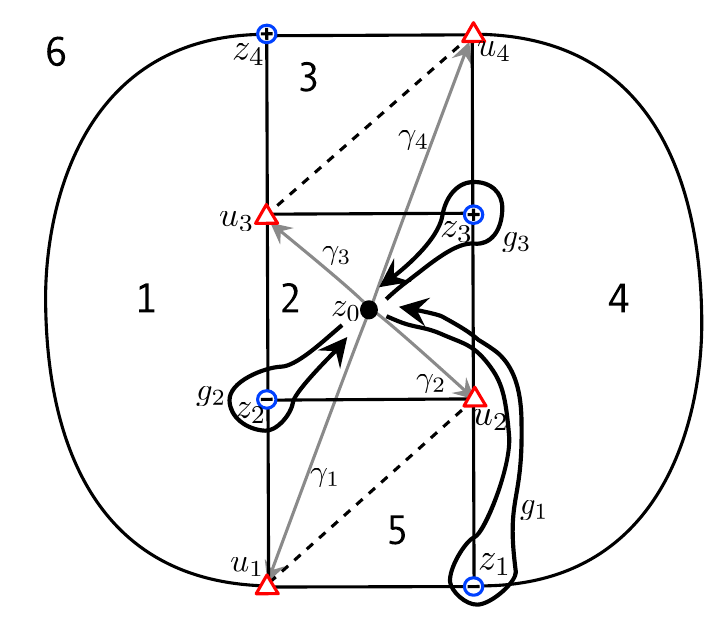}
\caption{Stokes graph of Type 1.}
\end{figure}

\begin{equation}
    g_1 = \begin{pmatrix}
        1 & -i w_2 \\
        0 & 1
    \end{pmatrix}\cdot \begin{pmatrix}
        1 & 0 \\ -i \frac{1}{w_2} & 1
    \end{pmatrix} \cdot \begin{pmatrix}
        1 & 0 \\ -i \frac{w_1}{w_2^2} & 1
    \end{pmatrix} \cdot  \begin{pmatrix}
        1 & 0 \\ -i \frac{1}{w_4n_3^2n_1^2} & 1
    \end{pmatrix}\begin{pmatrix}
        1 & i n_1^2 w_2 \\
        0 & 1
    \end{pmatrix}\cdot \begin{pmatrix}
        n_1 & 0 \\ 0 & \frac{1}{n_1}
    \end{pmatrix}
\end{equation}

\begin{equation}
    g_2 = \begin{pmatrix}
        1 & 0 \\ 
        -i\frac{1}{w_3} & 1
    \end{pmatrix} \cdot \begin{pmatrix}
        1 & 0 \\ 
        -i\frac{1}{w_1n_2^2} & 1
    \end{pmatrix} \cdot \begin{pmatrix}
        1 & 0 \\ 
        -i\frac{1}{w_2n_2^2} & 1
    \end{pmatrix} \cdot \begin{pmatrix}
        n_2 & 0 \\ 
        0 & \frac{1}{n_2}
    \end{pmatrix}
\end{equation}

\begin{equation}
    g_3 = \begin{pmatrix}
        1 & -iw_2 \\
        0 & 1
    \end{pmatrix} \cdot  \begin{pmatrix}
        1 & -i w_4 n_3^2 \\
        0 & 1
    \end{pmatrix} \cdot \begin{pmatrix}
        1 & -i w_3 n_3^2 \\
        0 & 1
    \end{pmatrix} \cdot \begin{pmatrix}
        n_3 & 0 \\
        0 & \frac{1}{n_3}
    \end{pmatrix}
\end{equation}

\begin{equation}
    A = \Tr(g_1g_2) = -\frac{{n_2} {w_{12}}}{{n_1} {w_{13}}}-\frac{{n_1}
   {n_2}}{{w_{12}}}-\frac{{w_{12}}}{{n_1} {n_2}}-\frac{{n_1}
   {n_2}}{{w_{13}}} -\frac{{n_2}^2 {n_4} {w_{12}}}{{n_3}
   {w_{13}}^2}-\frac{{n_2}^2 {n_4}}{{n_3} {w_{13}}}-\frac{{n_4}
   {w_{12}}}{{n_3} {w_{13}}}
\end{equation}

\begin{equation}
    B = \Tr(g_2g_3)=-\frac{w_{13} w_{12}}{n_1 n_2^2 n_4}-\frac{w_{12}}{n_2 n_3}-\frac{w_{12}}{n_1
   n_4}-\frac{n_2 w_{12}}{n_3 w_{13}}-\frac{n_3 w_{13}}{n_2}-\frac{w_{13}}{n_1 n_2^2
   n_4}-\frac{n_3 w_{13}}{n_2 w_{12}}
\end{equation}
\begin{equation}
    C = \Tr(g_1g_3) = -\frac{n_1 n_3 w_{13}}{w_{12}}-\frac{n_1 n_3 w_{13}}{w_{12}^2}-\frac{n_1
   n_3}{w_{12}}-\frac{w_{13}}{n_2 n_4}-\frac{w_{13}}{n_2 n_4 w_{12}}-\frac{n_2
   n_4}{w_{12}}-\frac{n_2 n_4}{w_{13}}
\end{equation}

\paragraph{Weak coupling.} $A$ has no leading dependence on the plumbing exponential \(w_{13}\), 
while \(B\) has leading \(w_{13}\)-dependence in the limit $q\to 0$.

\subsection{Type 2}
\begin{figure}[h]
\centering
\includegraphics[height=0.5\textwidth]{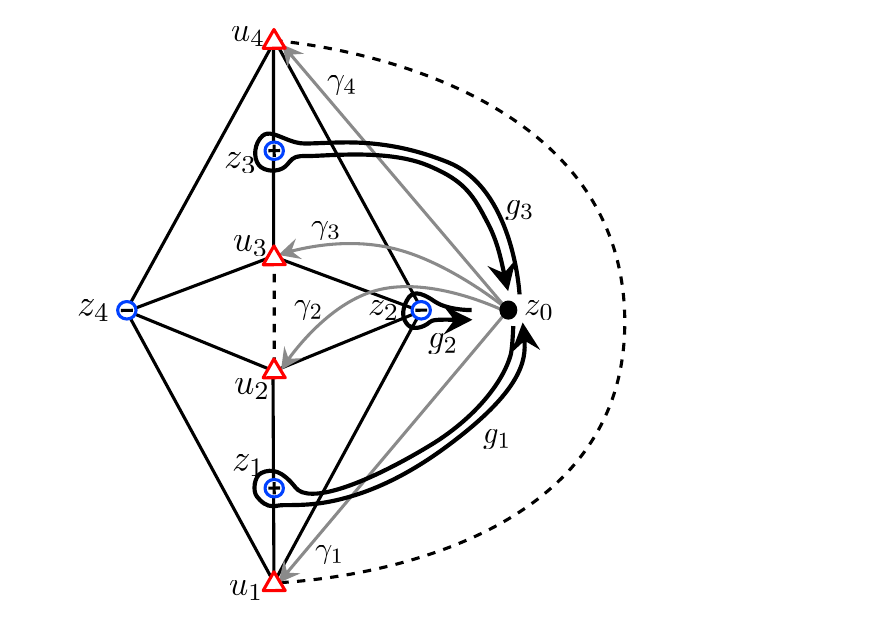}
\caption{Stokes graph of Type 2.}
\end{figure}

\begin{equation}
    g_1 = \begin{pmatrix}
        1 &0 \\
         i\frac{1}{w_1} & 1
    \end{pmatrix} \cdot  \begin{pmatrix}
        1 &  -i\frac{w_2}{n_2^2} \\
       0 & 1
    \end{pmatrix} \cdot \begin{pmatrix}
        1 & -i w_1 n_1^2 \\
        0 & 1
    \end{pmatrix} \cdot \begin{pmatrix}
        1 & 0 \\
        -i\frac{1}{w_1 n_1^2} & 1
    \end{pmatrix} \cdot \begin{pmatrix}
        n_1 & 0 \\
        0 & \frac{1}{n_1}
    \end{pmatrix}
\end{equation}
\begin{equation}
    g_2=\begin{pmatrix}
        1 & 0 \\
        -i \frac{1}{w_4} & 1
    \end{pmatrix} \cdot \begin{pmatrix}
        1 & 0 \\
        -i \frac{1}{w_3} & 1
    \end{pmatrix} \cdot \begin{pmatrix}
        1 & 0 \\
        -i \frac{1}{w_2} & 1
    \end{pmatrix} \cdot \begin{pmatrix}
        1 & 0 \\
        -i \frac{1}{w_1n_2^2} & 1
    \end{pmatrix} \cdot \begin{pmatrix}
        n_2 & 0 \\
        0 & \frac{1}{n_2}
    \end{pmatrix}
\end{equation}
\begin{equation}
    g_3 = \begin{pmatrix}
        1 & 0 \\
        -i \frac{1}{w_4} & 1
    \end{pmatrix} \cdot \begin{pmatrix}
        1 & -i{w_4} \\
        0 & 1
    \end{pmatrix} \cdot \begin{pmatrix}
        1 & -i{w_3}n_3^2 \\
        0 & 1
    \end{pmatrix} \cdot \begin{pmatrix}
        1 & 0 \\
        i\frac{1}{w_4n_3^2} & 1
    \end{pmatrix} \begin{pmatrix}
        n_3 & 0 \\
        0 & \frac{1}{n_3}
    \end{pmatrix}
\end{equation}

\begin{equation}
    A=-\frac{w_{12}^2}{n_1^2 n_2^2 n_3 n_4 w_{13}}-\frac{w_{12}}{n_1 n_2}-\frac{w_{12}}{n_1 n_2
   w_{13}}-\frac{w_{12}}{n_3 n_4 w_{13}}-\frac{n_1 n_2}{w_{13}}-\frac{n_1 n_2}{w_{12}}
\end{equation}
\begin{equation}
    B = -\frac{n_1 n_4 n_2^2 w_{13}}{w_{12}^2}-\frac{n_1 n_4 n_2^2}{w_{12}}-\frac{n_3 n_2
   w_{13}}{w_{12}}-\frac{n_1 n_4 w_{13}}{w_{12}}-\frac{n_3 w_{13}}{n_2}-\frac{w_{12}}{n_1
   n_4 n_2^2}
\end{equation}
\begin{multline}
    C = -\frac{w_{12}^3}{n_1^3 n_2^4 n_3 n_4^2 w_{13}}-\frac{2 w_{12}^2}{n_1^2 n_2^3
   n_4}-\frac{w_{12}^2}{n_1^2 n_2^3 n_4 w_{13}}-\frac{w_{12}^2}{n_1 n_2^2 n_3 n_4^2
   w_{13}}-\frac{n_3 w_{12}}{n_1 n_2^2}-\\-\frac{n_3 w_{13} w_{12}}{n_1
   n_2^2}-\frac{w_{12}}{n_1 n_2^2 n_3}-\frac{w_{12}}{n_1^2 n_2 n_4}-\frac{w_{12}}{n_2
   n_4}-\frac{w_{12}}{n_2 n_4 w_{13}}-\frac{n_3 w_{13}}{n_1}-\frac{n_4
   w_{13}}{n_2}-\frac{n_2 n_4 w_{13}}{w_{12}}
\end{multline}

\paragraph{Weak coupling.} $A$ has no leading dependence on the plumbing exponential \(w_{13}\), 
while \(B\) has leading \(w_{13}\)-dependence in the limit $q\to 0$.

\subsection{Type 3}
\begin{figure}[h]
\centering
\includegraphics[height=0.5\textwidth]{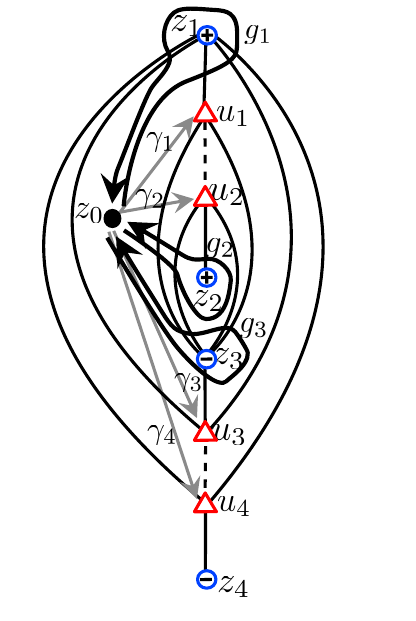}
\caption{Stokes graph of Type 3.}
\end{figure}

\begin{multline}
    g_1 = \begin{pmatrix}
        1 & -iw_1 \\ 
        0 & 1
    \end{pmatrix} \cdot\begin{pmatrix}
        1 & -i \frac{w_3w_1^2}{n_3^2 n_2^2w_2^2} \\ 
        0 & 1
    \end{pmatrix}\cdot\begin{pmatrix}
        1 & -i \frac{w_3^2w_1^2}{w_4w_2^2n_3^2n_2^2} \\ 
        0 & 1
    \end{pmatrix}\cdot \\ \cdot\begin{pmatrix}
        1 & -iw_4n_1^2 \\ 
        0 & 1
    \end{pmatrix}\cdot\begin{pmatrix}
        1 & -iw_3 n_1^2 \\ 
        0 & 1
    \end{pmatrix}\cdot\begin{pmatrix}
        n_1 & 0 \\ 
        0 & \frac{1}{n_1}
    \end{pmatrix}
\end{multline}

\begin{multline}
    g_2 = \begin{pmatrix}
        1 & 0\\
        i \frac{1}{w_1} & 1
    \end{pmatrix} \cdot \begin{pmatrix}
        1 & 0\\
        i \frac{1}{w_2} & 1
    \end{pmatrix} \cdot \begin{pmatrix}
        1 & -i w_2 n_2^2\\
        0 & 1
    \end{pmatrix} \cdot \begin{pmatrix}
        1 & 0\\
        -i \frac{1}{w_2 n_2^2} & 1
    \end{pmatrix} \cdot \begin{pmatrix}
        1 & 0\\
        i \frac{1}{w_1n_2^2} & 1
    \end{pmatrix} \cdot\begin{pmatrix}
        n_2 & 0 \\ 
        0 & \frac{1}{n_2}
    \end{pmatrix} 
\end{multline}

\begin{multline}
g_3 = \begin{pmatrix}
        1 & 0\\
        -i \frac{1}{w_3} & 1
    \end{pmatrix} \cdot \begin{pmatrix}
        1 & 0\\
        -i \frac{w_1}{w_2^2n_2^2n_3^2} & 1
    \end{pmatrix} \cdot \begin{pmatrix}
        1 & 0\\
        -i \frac{1}{w_2n_2^2n_3^2} & 1
    \end{pmatrix} \cdot \\ \cdot \begin{pmatrix}
        1 & 0\\
        -i \frac{1}{w_2n_3^2} & 1
    \end{pmatrix} \cdot \begin{pmatrix}
        1 & 0\\
        -i \frac{1}{w_1n_3^2} & 1
    \end{pmatrix} \cdot \begin{pmatrix}
        n_3 & 0 \\ 
        0 & \frac{1}{n_3}
    \end{pmatrix} 
\end{multline}

\begin{multline}
    A = -n_1 n_2 w_{12}-n_1 n_2 w_{13} w_{12}-n_1 n_2 w_{13}-\frac{n_4
   w_{13}}{n_3}-\frac{n_1 w_{13}}{n_2}-\frac{w_{13}}{n_3 n_4}-\\-\frac{n_4
   w_{13}}{n_2^2 n_3 w_{12}}-\frac{n_4 w_{13}}{n_3 w_{12}}-\frac{n_1
   w_{13}}{n_2 w_{12}}-\frac{w_{13}}{n_2^2 n_3 n_4 w_{12}}-\frac{w_{13}}{n_3
   n_4 w_{12}}-\frac{w_{13}}{n_1 n_2 n_3^2 w_{12}}-\frac{1}{n_1 n_2
   w_{12}}-\\-\frac{n_4 w_{13}}{n_2^2 n_3 w_{12}^2}-\frac{w_{13}}{n_2^2 n_3 n_4
   w_{12}^2}-\frac{w_{13}}{n_1 n_2 n_3^2 w_{12}^2}-\frac{w_{13}}{n_1 n_2^3
   n_3^2 w_{12}^2}-\frac{w_{13}}{n_1 n_2^3 n_3^2 w_{12}^3}
\end{multline}

\begin{equation}
    B = -n_2 n_3 w_{12}-\frac{n_2 n_3 w_{12}}{w_{13}}-\frac{1}{n_2 n_3 w_{12}}
\end{equation}
\begin{multline}
C=-\frac{n_3}{n_1 w_{13}}-\frac{n_4}{n_2 w_{12}}-\frac{1}{n_2 n_4
   w_{12}}-\frac{n_1 w_{13}}{n_3}-\frac{n_1 w_{13}}{n_3 w_{12}}-\frac{n_1
   w_{13}}{n_2^2 n_3 w_{12}}-\frac{n_1 w_{13}}{n_2^2 n_3
   w_{12}^2}-\\-\frac{1}{n_1 n_3 w_{12}}-\frac{1}{n_1 n_2^2 n_3
   w_{12}}-\frac{2}{n_1 n_2^2 n_3 w_{12}^2}-\frac{n_4 w_{13}}{n_2 n_3^2
   w_{12}}-\frac{w_{13}}{n_2 n_4 n_3^2 w_{12}}-\frac{n_4 w_{13}}{n_2 n_3^2
   w_{12}^2}-\\-\frac{n_4 w_{13}}{n_2^3 n_3^2 w_{12}^2}-\frac{w_{13}}{n_2 n_4
   n_3^2 w_{12}^2}-\frac{w_{13}}{n_2^3 n_4 n_3^2 w_{12}^2}-\frac{n_4
   w_{13}}{n_2^3 n_3^2 w_{12}^3}-\frac{w_{13}}{n_2^3 n_4 n_3^2
   w_{12}^3}-\\-\frac{w_{13}}{n_1 n_2^2 n_3^3 w_{12}^2}-\frac{w_{13}}{n_1 n_2^2
   n_3^3 w_{12}^3}-\frac{w_{13}}{n_1 n_2^4 n_3^3 w_{12}^3}-\frac{w_{13}}{n_1
   n_2^4 n_3^3 w_{12}^4}    
\end{multline}
\paragraph{Weak coupling.} $A$ has no leading dependence on the plumbing exponential \(w_{13}^{-1}\), 
while \(B\) has leading \(w_{13}^{-1}\)-dependence in the limit $q\to 0$.

\subsection{Type 4}

\begin{figure}[h]
\centering
\includegraphics[height=0.5\textwidth]{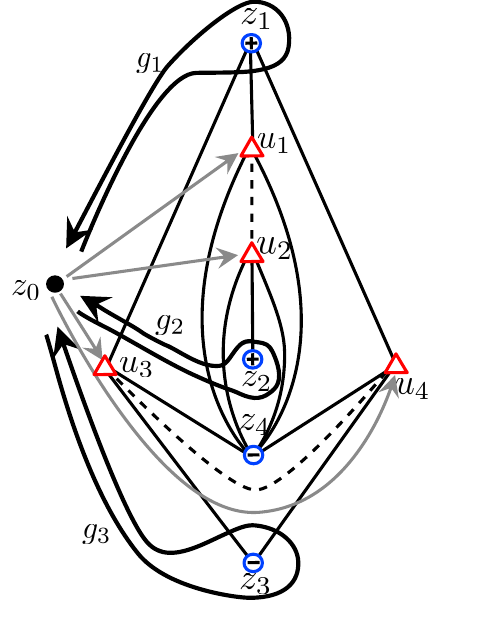}
\caption{Stokes graph of Type 4.}
\end{figure}

\begin{equation}
    g_1 = \begin{pmatrix}
        1 & -i w_3 \\
        0 & 1
    \end{pmatrix} \cdot \begin{pmatrix}
        1 & -i w_1 \\
        0 & 1
    \end{pmatrix} \cdot \begin{pmatrix}
        1 & -i w_4 n_1^2 n_3^2 \\
        0 & 1
    \end{pmatrix} \cdot \begin{pmatrix}
        n_1 & 0 \\
        0 & \frac{1}{n_1}
    \end{pmatrix}
\end{equation}
\begin{multline}
    g_2 = \begin{pmatrix}
        1 & -i w_3 \\
        0 & 1
    \end{pmatrix} \cdot \begin{pmatrix}
        1 & 0 \\
        i \frac{1}{w_1} & 1
    \end{pmatrix} \cdot \begin{pmatrix}
        1 & 0 \\
        i \frac{1}{w_2} & 1
    \end{pmatrix} \cdot \begin{pmatrix}
        1 & -i w_2 n_2^2 \\
        0 & 1
    \end{pmatrix} \cdot \\ \cdot \begin{pmatrix}
        1 & 0 \\
        -i \frac{1}{w_2n_2^2} & 1
    \end{pmatrix}\cdot \begin{pmatrix}
        1 & 0 \\
        -i \frac{1}{w_1n_2^2} & 1
    \end{pmatrix}\cdot \begin{pmatrix}
        1 & i w_3 n_2^2 \\
        0 & 1
    \end{pmatrix}\cdot \begin{pmatrix}
        n_2 & 0 \\
        0 & \frac{1}{n_2}
    \end{pmatrix}
\end{multline}

\begin{equation}
    g_3 = \begin{pmatrix}
        1 & 0 \\
        -i \frac{1}{w_4n_3^2} & 1
    \end{pmatrix} \cdot \begin{pmatrix}
        1 & 0 \\
        -i \frac{1}{w_3n_3^2} & 1
    \end{pmatrix} \cdot \begin{pmatrix}
        n_3 & 0 \\
        0 & \frac{1}{n_3}
    \end{pmatrix}
\end{equation}

\begin{multline}
    A=-n_1 n_2 w_{12}-n_1 n_2 w_{13} w_{12}-n_1 n_2 w_{13}-\frac{n_1 w_{13}}{n_2}-\frac{n_3
   w_{13}}{n_4}-\frac{n_1 w_{13}}{n_2 w_{12}}- \\ -\frac{n_3 w_{13}}{n_2^2 n_4
   w_{12}}- \frac{n_3 w_{13}}{n_4 w_{12}}-\frac{1}{n_1 n_2 w_{12}}-\frac{n_3 w_{13}}{n_2^2
   n_4 w_{12}^2}
\end{multline}
\begin{multline}
 B = -2 n_1 n_4 n_2^2 w_{12}^2-n_1 n_4 n_2^2 w_{12}-n_1 n_4 n_2^2 w_{12}^2 w_{13}-n_1 n_4
   n_2^2 w_{12} w_{13}-\frac{n_1 n_4 n_2^2 w_{12}^2}{w_{13}}-\\ -n_3 n_2 w_{12}- \frac{n_2
   w_{12}}{n_3}-n_3 n_2 w_{13}-n_3 n_2 w_{12} w_{13}-\frac{n_2 w_{12}}{n_3 w_{13}}-n_1
   n_4 w_{12}- \\ -n_1 n_4 w_{13}-n_1 n_4 w_{12} w_{13}-\frac{n_3 w_{13}}{n_2}-\frac{n_3
   w_{13}}{n_2 w_{12}}    
\end{multline}
\begin{equation}
    C=-n_2 n_4 w_{12}-\frac{n_2 n_4 w_{12}}{w_{13}}-\frac{1}{n_1 n_3 w_{13}}-\frac{1}{n_2 n_4
   w_{12}}
\end{equation}
\paragraph{Weak coupling.} $A$ has no leading dependence on the plumbing exponential \(w_{13}\), 
while \(B\) has leading \(w_{13}\)-dependence in the limit $q\to 0$.

\subsection{Type 5}

\begin{figure}[h]
\centering
\includegraphics[height=0.5\textwidth]{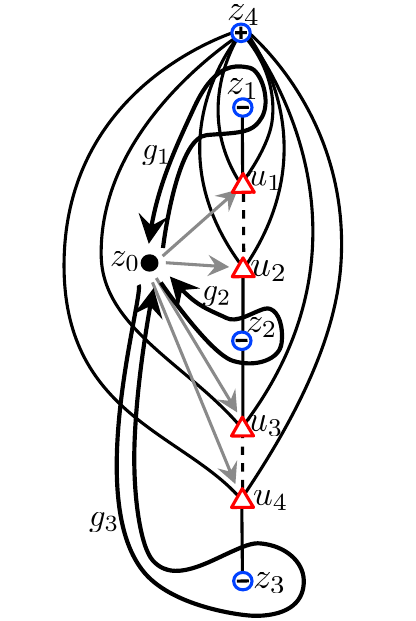}
\caption{Stokes graph of Type 5.}
\end{figure}

\begin{equation}
    g_1 = \begin{pmatrix}
        1 & -i w_2 \\
        0 & 1
    \end{pmatrix}\cdot\begin{pmatrix}
        1 & -i w_1 \\
        0 & 1
    \end{pmatrix}\cdot\begin{pmatrix}
        1 & 0 \\
        -i \frac{1}{w_1} & 1
    \end{pmatrix}\cdot\begin{pmatrix}
        1 & i w_1 n_1^2 \\
        0 & 1
    \end{pmatrix}\cdot\begin{pmatrix}
        1 & i w_2 n_1^2 \\
        0 & 1
    \end{pmatrix}\cdot\begin{pmatrix}
        n_1 & 0 \\
        0 & \frac{1}{n_1}
    \end{pmatrix}
\end{equation}
\begin{equation}
    g_2 = \begin{pmatrix}
        1 & 0\\
        -i \frac{1}{w_3} & 1
    \end{pmatrix}\cdot \begin{pmatrix}
        1 & 0 \\
        -i \frac{1}{w_2 n_2^2} & 1
    \end{pmatrix}\cdot \begin{pmatrix}
        n_2 & 0 \\
        0 & \frac{1}{n_2}
    \end{pmatrix}
\end{equation}
\begin{equation}
    g_3 = \begin{pmatrix}
        1 & iw_3 \\
        0 & 1
    \end{pmatrix} \cdot\begin{pmatrix}
        1 & iw_4 \\
        0 & 1
    \end{pmatrix} \cdot\begin{pmatrix}
        1 & 0 \\
        -i\frac{1}{w_4n_3^2} & 1
    \end{pmatrix} \cdot\begin{pmatrix}
        1 & -iw_4 n_3^2 \\
        0 & 1
    \end{pmatrix} \cdot\begin{pmatrix}
        1 & -iw_3n_3^2 \\
        0 & 1
    \end{pmatrix} \cdot\begin{pmatrix}
        n_3 & 0 \\
        0 & \frac{1}{n_3}
    \end{pmatrix}
\end{equation}

\begin{equation}
    A = -\frac{n_1 n_2 w_{12}^2}{w_{13}}-n_1 n_2 w_{12}-\frac{n_1 n_2 w_{12}}{w_{13}}-\frac{n_2
   w_{12}}{n_1 w_{13}}-\frac{n_2}{n_1 w_{13}}-\frac{1}{n_1 n_2 w_{12}}
\end{equation}
\begin{equation}
    B = -n_1 n_4 w_{12}-n_1 n_4 w_{13}-\frac{n_3 w_{13}}{n_2 w_{12}}-\frac{w_{13}}{n_2 n_3
   w_{12}}-\frac{1}{n_1 n_4 w_{12}}-\frac{w_{13}}{n_1 n_2^2 n_4 w_{12}^2}
\end{equation}
\begin{multline}
C=-\frac{n_1^2 n_2 n_4 w_{12}^3}{w_{13}}-2 n_1^2 n_2 n_4 w_{12}^2-\frac{n_1^2 n_2 n_4
   w_{12}^2}{w_{13}}-\frac{n_2 n_4 w_{12}^2}{w_{13}}-n_1 n_3 w_{12}-n_1^2 n_2 n_4
   w_{12}-\\-n_2 n_4 w_{12}-n_1^2 n_2 n_4 w_{13} w_{12}-\frac{n_1 w_{12}}{n_3}-\frac{n_2 n_4
   w_{12}}{w_{13}}-n_1 n_3 w_{13}-\frac{n_1 w_{13}}{n_3}-\frac{w_{13}}{n_2 n_4 w_{12}}    
\end{multline}
\paragraph{Weak coupling.} $A$ has no leading dependence on the plumbing exponential \(w_{13}\), 
while \(B\) has leading \(w_{13}\)-dependence in the limit $q\to 0$.

\subsection{Type 6}

\begin{figure}[h]
\centering
\includegraphics[height=0.5\textwidth]{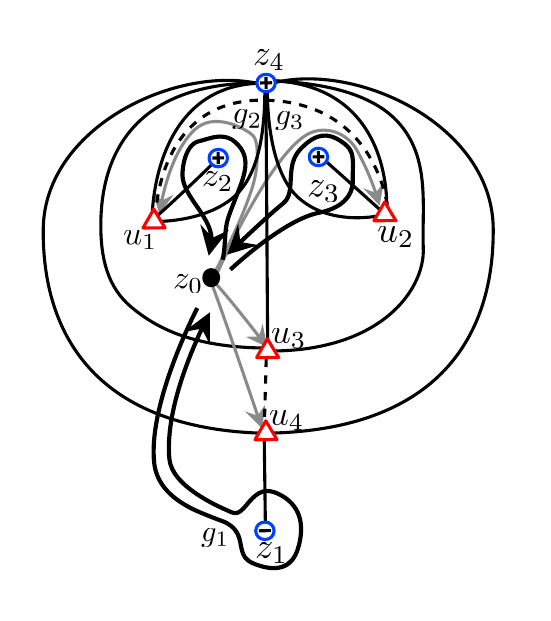}
\caption{Stokes graph of Type 6.}
\end{figure}

\begin{equation}
    g_1 = \begin{pmatrix}
        1 & iw_3 \\
        0 & 1
    \end{pmatrix} \cdot \begin{pmatrix}
        1 & iw_4 \\
        0 & 1
    \end{pmatrix} \cdot\begin{pmatrix}
        1 & 0 \\
        -i \frac{1}{w_4n_1^2} & 1
    \end{pmatrix} \cdot\begin{pmatrix}
        1 & -iw_4 n_1^2 \\
        0 & 1
    \end{pmatrix} \cdot\begin{pmatrix}
        1 & -iw_3 n_1^2 \\
        0 & 1
    \end{pmatrix} \cdot \begin{pmatrix}
        n_1 & 0 \\
        0 & \frac{1}{n_1}
    \end{pmatrix}
\end{equation}
\begin{equation}
    g_2 = \begin{pmatrix}
        1 & 0 \\
        i \frac{n_2^2}{w_1} & 1
    \end{pmatrix} \cdot \begin{pmatrix}
        1 & -iw_1 \\
        0 & 1
    \end{pmatrix} \cdot \begin{pmatrix}
        1 & 0 \\
        -i \frac{1}{w_1} & 1
    \end{pmatrix} \cdot \begin{pmatrix}
        n_2 & 0 \\
        0 & \frac{1}{n_2}
    \end{pmatrix}
\end{equation}
\begin{equation}
    g_3 = \begin{pmatrix}
        1 & 0 \\
        -i\frac{1}{w_3} & 1
    \end{pmatrix} \cdot \begin{pmatrix}
        1 & 0 \\
        -i\frac{1}{w_2 n_3^2} & 1
    \end{pmatrix} \cdot \begin{pmatrix}
        1 & -i w_2 n_3^2 \\
        0 & 1
    \end{pmatrix} \cdot \begin{pmatrix}
        1 & 0 \\
        \frac{i}{w_2 n_3^4} & 1
    \end{pmatrix} \cdot \begin{pmatrix}
        1 & 0 \\
        \frac{i}{w_3 n_3^2} & 1
    \end{pmatrix}\cdot \begin{pmatrix}
        n_3 & 0 \\
        0 & \frac{1}{n_3}
    \end{pmatrix}
\end{equation}

\begin{equation}
    A = -n_3 n_4 n_2^2 w_{12}-n_3 n_4 n_2^2 w_{12} w_{13}-n_1 n_2 w_{13}-\frac{n_2
   w_{13}}{n_1}-n_3 n_4 w_{12}-\frac{w_{13}}{n_3 n_4 w_{12}}-\frac{n_3 n_4
   w_{12}}{w_{13}}
\end{equation}
\begin{equation}
    B = -n_2 n_3 w_{12}-\frac{n_2 n_3 w_{12}}{w_{13}}-\frac{n_3 w_{12}}{n_2
   w_{13}}-\frac{n_3}{n_2 w_{13}}-\frac{n_3 w_{12}}{n_2 w_{13}^2}-\frac{1}{n_2 n_3
   w_{12}}-\frac{1}{n_2 n_3 w_{13}}
\end{equation}
\begin{equation}
    C = -n_2 n_4 w_{13}-\frac{n_1 w_{13}}{n_3 w_{12}}-\frac{w_{13}}{n_1 n_3
   w_{12}}-\frac{w_{13}}{n_2 n_3^2 n_4 w_{12}^2}-\frac{1}{n_2 n_4 w_{12}}-\frac{1}{n_2
   n_3^2 n_4 w_{12}}-\frac{1}{n_2 n_4 w_{13}}
\end{equation}

\paragraph{Weak coupling.} Note that there is no limit $q \to 0$ in which $z_1, z_2$ and $u_1, u_2$ end up in the inner region. This explains why $A$ has terms both proportional to $w_{13}$ and its inverse. A valid weak coupling limit in this case corresponds to $z_2 \to z_3$, with $u_1$ and $u_2$ ending up in the same region. That explains why $B$ is a Taylor series in $w_{13}^{-1}$. $A$, on the other hand, has terms proportional to $w_{13}$, which dominate in the $z_2 \to z_3$ limit.

\bibliographystyle{unsrtnat}
\bibliography{references.bib}

\end{document}